\documentclass[letterpaper, 11pt]{article}
\usepackage[margin=1in]{geometry}
\usepackage{natbib}
\usepackage{palatino}

\usepackage{amsfonts,amsmath,amsthm,amssymb}
\usepackage{tikz}
\usepackage{wrapfig}
\usepackage{xcolor}
\usepackage{authblk}
\usepackage{url}

\usepackage{tikz}
\usepackage[linesnumbered,ruled]{algorithm2e}
\usepackage{enumitem}
\usepackage{mathtools}
\usepackage[font=small]{subcaption}

\usepackage{hyperref}

\newtheorem{theorem}{Theorem}[section]

\newtheorem{lemma}[theorem]{Lemma}
\newtheorem{corollary}[theorem]{Corollary}

\newtheorem{definition}[theorem]{Definition}

\newtheorem{example}[theorem]{Example}

\definecolor{black}{RGB}{0,0,0}
\definecolor{orange}{RGB}{230,159,0}
\definecolor{skyblue}{RGB}{86,180,233}
\definecolor{bluishgreen}{RGB}{0,158,115}
\definecolor{yellow}{RGB}{240,228,66}
\definecolor{blue}{RGB}{0,114,178}
\definecolor{vermillion}{RGB}{213,94,0}
\definecolor{reddishpurple}{RGB}{204,121,167}
\definecolor{cugold}{RGB}{207,184,124}

\newcommand{\projection}[2]{\mathrm{proj}_{#1}\left(#2\right)}
\newcommand{\interior}[1]{\mathrm{int}
{\left(#1\right)}}

\newcommand{\conv}[1]{\mathrm{conv}{\left(#1\right)}}
\newcommand{\Q}{\mathbb{Q}}
\newcommand{\R}{\mathbb{R}}
\DeclareMathOperator*{\argmax}{arg\,max}

\title{Linear Production Games with Non-transferable Utilities}
\author[1]{J. Carlos Mart\'inez Mori}
\affil[1]{Department of Mathematical and Statistical Sciences, University of Colorado Denver, carlos.martinezmori@ucdenver.edu}

\author[2]{Alejandro Toriello}
\affil[2]{H. Milton Stewart School of Industrial and Systems Engineering, Georgia Institute of Technology, atoriello@isye.gatech.edu}

\begin{document}
\maketitle

\begin{abstract}
We introduce non-transferable utility linear production (NTU LP) games, a non-transferable utility analogue of classical linear production games, as a framework for the study of cooperative behavior in the production or establishment of public goods with pooled resources.
NTU LP games combine the game-theoretic tensions inherent in public decision-making with the modeling flexibility of linear programming.
We derive structural properties regarding the non-emptiness, representability and complexity of the core, a solution concept that models the viability of cooperation.
In particular, we provide fairly general sufficient conditions under which the core of an NTU LP game is guaranteed to be non-empty, prove that determining membership in the core is co-NP-complete, and develop a cutting plane algorithm to optimize various social welfare objectives subject to core membership.
We apply these results in a data-driven case study on service plan optimization for the Chicago bus system.
As our study illustrates, cooperation is necessary for the successful deployment of transportation service plans and similar public goods, but it may also have adverse or counterintuitive distributive implications.
\end{abstract}

\section{Introduction}
\label{sec: introduction}

In this paper, we introduce \emph{non-transferable utility linear production} (NTU LP) games, a non-transferable utility analogue of \citet{owen1975core}'s linear production games.
In public decision-making environments, it can be challenging to reconcile strategic goals (e.g., system efficiency, equity, sustainability) with individual stakeholder interests.
This is exacerbated by social inequality and the ensuing power dynamics, which can sway public decision-making in favor of particular groups.
Our goal in this work is to understand these power dynamics, and to navigate them effectively in the design of public goods that advance strategic goals for policy and infrastructure.

We are particularly concerned with transportation planning, which involves the design and implementation of systems for the movement of goods and people.
In the public sector, it is a collaborative process that involves a wide variety of stakeholders: government, businesses, advocates, and users.
As such, the success of a transportation planning process is critically dependent on its ability to garner sufficient support from each of these groups.
However, as we illustrate in Section~\ref{sec: motivating example}, disparate stakeholder interests on the use of a common resource pool can bring about game-theoretic tensions that complicate decision-making.

To address this kind of difficulty, we combine NTU with LP.
Unlike private goods, which are excludable and whose distribution is typically determined by the principle of free exchange, public goods are non-excludable.
Moreover, certain public goods have a distinct social meaning, as does the transportation good defined as accessibility, which justifies their allocation through a distributive principle different from free exchange~\citep{walzer1983spheres,martens2012justice}.
This makes NTU especially relevant for studying the strategic elements that impact the distribution of a public good.
And in analogy to \citet{owen1975core}, combining it with LP brings in the modeling flexibility and algorithmic techniques of mathematical programming.

\subsection{Motivating Example: A Ridership versus Coverage Dilemma}
\label{sec: motivating example}

Consider a public transit agency with three riders, $1$, $2$, and $3$, and two bus lines, \textcolor{orange}{A} and \textcolor{blue}{B}, as depicted in Figure~\ref{fig: dilemma}.
Riders $2$ and $3$ live and work downtown, and can be served by either line.
On the other hand, rider $1$ lives and works in suburbs far from downtown (in opposite ends of the city), and can only be served by the \textcolor{orange}{A} line.
Because of its coverage-oriented design, the \textcolor{orange}{A} line is significantly longer than the \textcolor{blue}{B} line.
Therefore, for any fixed target frequency, it is much costlier to operate the \textcolor{orange}{A} line than the \textcolor{blue}{B} line.

\begin{figure}[ht]
    \centering
    \caption{
        The \textcolor{orange}{A} line is long and shown in orange, whereas the \textcolor{blue}{B} line is short and shown in blue.
    }
    \label{fig: dilemma}
    \includegraphics[trim={5cm 12.5cm 5cm 12.5cm},clip,width=0.6\linewidth]{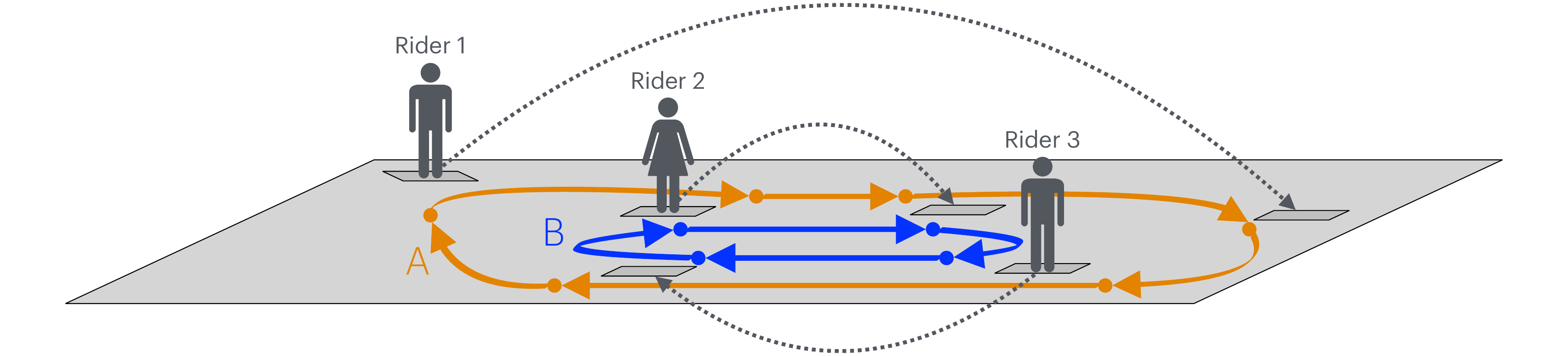}
\end{figure}

The agency funds its operations strictly through fare-box recovery\textemdash no external funding is available.
Each rider is charged a flat fare of \$1.00.
Thus, the agency must decide how to distribute its \$3.00 budget among its two lines.
One option is to not operate the \textcolor{blue}{B} line at all, and to dedicate its entire budget to the operation of the \textcolor{orange}{A} line.
All riders would be served in this scenario, albeit possibly at an unsatisfactory frequency.
This is nevertheless the best-case scenario for rider $1$, whose service would be partly cross-subsidized by the participation of riders $2$ and $3$.

However, this service plan might not be a viable alternative after all.
In particular, both riders $2$ and $3$ could be served by the \textcolor{blue}{B} line alone, and with their joint \$2.00 budget they could potentially operate it independently at a high frequency, which is crucial in transit adoption \citep{walker2024human}.
Therefore, together they could form a credible threat against any service plan that goes too far in neglecting the operation of the \textcolor{blue}{B} line.

Why would the transit agency be interested in operating the A line alone when its main ridership base, i.e., riders like 2 and 3, might perceive this as inefficient, unpopular, and unfair?
The answer to this question depends on whether we evaluate transit by its ``productivity'' or by its role as a social service \citep{walker2024human}.
An agency that focuses on compatible riders in high-density areas unlocks operational economies of scale and boosts the economic activity of its service region.
Thus, one could justify poor service for riders in low-density areas, like rider 1 in the example, based on the imbalance between what it costs to serve them and how much they contribute to the agency's budget.
Similarly, one could justify increasing the fare for these riders through a distance-based calculation.
These arguments are much less compelling if the explicit role of a transit agency is that of a social service; one that covers the basic transportation needs of all individuals but especially of those who are disadvantaged, for reasons including but not limited to their car-ownership status, ability to drive, poverty, and disabilities.
In this case, a transit agency is judged by how good a service it provides to those who most need it.
In practice, transit agencies typically develop their service plan by apportioning different percentages of their total budget toward these two, often conflicting design goals.

The example highlights tensions present in many real-world transit systems.
In many college towns across the United States, local transit agencies struggle to balance the operation of community-oriented routes alongside service agreements with university partners, as in Ithaca, NY \citep{dougherty2024tcat}.
In other college towns, this has resulted in the establishment of independent, university-operated shuttle systems, e.g.\ Tuscaloosa, AL \citep{caddell2005plans}, which benefit university-affiliated riders but represent a missed funding opportunity for local transit authorities.
In large cities, this tension has been reflected in failed transit-funding referenda, e.g., Nashville, TN \citep{transit2019derailed} or Atlanta, GA \citep{robinson2024cobb}, with lasting consequences for policy making due to the difficulty of revisiting failed proposals.
In certain extreme cases, broader tensions over municipal funding have led to entire areas threatening to or successfully splitting off, e.g., the suburb of Buckhead, which came close to seceding from the city of Atlanta \citep{mock2023how}, and the city of St.\ George, LA, which separated from Baton Rouge \citep{rojas2024louisiana}.
The social implications of this problem have only grown after decades of increasingly suburbanized poverty \citep{howell2014racial} and its negative effects on transportation equity \citep{kramer2018unaffordable}, especially following the COVID-19 pandemic \citep{kneebone2023post}.

\subsection{Contributions}
\label{sec: contributions}

The tension illustrated in Section \ref{sec: motivating example} speaks to a fundamental difficulty with designing public goods, including but not limited to transit, that take on a social service mission while trying to maintain broad popular support.
Therefore, in this paper, we consider the design of such goods from a general mathematical programming perspective.
We introduce \emph{non-transferable utility linear production} (NTU LP) games, a non-transferable utility analogue of \citet{owen1975core}'s linear production games, as a framework for the study of cooperative behavior in the production or establishment of public goods with pooled resources.
Our contributions are as follows:
\begin{enumerate}
    \item 
    We derive structural properties regarding the non-emptiness, representability, and complexity of the core, one of the main solution concepts used to determine the viability of cooperation.
    
    \item
    We develop and implement a cutting plane algorithm that accommodates the optimization of various notions of social welfare subject to core membership, which represents cooperative stability constraints.
    
    \item 
    We illustrate the potentially adverse and counterintuitive distributive implications of cooperation through a data-driven case study motivated by the example in Section~\ref{sec: motivating example}.
\end{enumerate}
Mitigating such consequences further motivates the outlook of this research.

The remainder of this paper is organized as follows.
In Section~\ref{sec: related work} we position our contributions within the existing literature.
In Section~\ref{sec: planning and cooperation} we introduce our framework and derive our analytical and algorithmic results.
In Section~\ref{sec: frequency setting in the chicago bus system} we present our data-driven implementation.
Lastly, we conclude in Section~\ref{sec: conclusion} with a discussion of future work.

\section{Related Work}
\label{sec: related work}

Cooperative game theory studies cooperative behavior among disparate agents who can nevertheless negotiate binding agreements.
A common assumption in much (but not all) of this theory is that of transferable utilities (TU): there is some divisible, tradeable commodity\textemdash typically money\textemdash that all players value linearly.
In this way, the outcomes of a TU game can be measured with respect to the total amount of utility they produce.
Formally, a (coalitional) TU game on a set of players $N$ is specified by a characteristic function $v: 2^N \rightarrow \mathbb{R}$ with $v(\emptyset) = 0$.
For any \emph{coalition} $\emptyset \neq S \subseteq N$, $v(S)$ represents the total amount of utility available to the members of $S$ through their exclusive and coordinated action;
\citet{vonneumann1944theory} derived TU games in coalitional form from TU games in strategic form.

The main objects of study in TU games are the payoff vectors.
A solution concept associates to each TU game a subset of payoff vectors that satisfy a given set of axioms; these underpin the credibility of a payoff vector as a potential outcome of the game.
Different combinations of axioms constitute different solution concepts, many of which are thoroughly detailed in \citet{peleg2007introduction}.
In this work we adopt the \emph{core} \citep{gillies59} as our solution concept, arguably the most widely used solution concept in operations research (OR).
A payoff vector $u \in \mathbb{R}^N$ is feasible if $\sum_{i \in N} u_i \leq v(N)$.
A coalition $\emptyset \neq S \subseteq N$ can improve on $u$ if and only if $v(S) > \sum_{i \in S} u_i$.
Finally, $u$ is said to be in the core if it is feasible and stable, that is $\sum_{i \in S} u_i \geq v(S)$ for all $\emptyset \neq S \subseteq N$.
If the core is non-empty, this suggests that cooperation among all players in $N$ is plausible.
\citet{bondareva1963some} and \citet{shapley1967balanced} characterized TU games with non-empty cores:
a collection of coalitions $\mathcal{S} \subseteq 2^N \setminus \{\emptyset\}$ is balanced if there exist non-negative weights $\lambda^S \geq 0$ for $S \in \mathcal{S}$ such that $\sum_{S \in \mathcal{S}: i \in S} \lambda^S = 1$ for all $i \in N$.
They termed a TU game $v$ as balanced if $\sum_{S \in \mathcal{S}} \lambda^S v(S) \leq v(N)$ for every balanced collection $\mathcal{S}$ with balancing weights $\{\lambda^S\}_{S \in \mathcal{S}}$, and showed that a TU game has a non-empty core if and only if it is balanced.

\citet{owen1975core} introduced linear production (LP) games, a special class of TU games derived from linear programs, and therefore one with a broad range of applications in OR.
In these games, players additively pool their individually endowed resources to produce goods that can be sold at given prices.
Therefore, it is in the interest of any coalition to produce an assortment of goods that maximizes its sales revenue.
Owen showed that these games are balanced, which by the results of \citet{bondareva1963some} and \citet{shapley1967balanced} implies that their core is non-empty.
Crucially, Owen also recovered this result algorithmically through linear programming duality\textemdash thus unlocking its broad practical impact.
However, he also showed that not all points in the core of a TU LP game can be obtained in this way.
In fact, \citet{fang2001membership} showed that, in general, testing membership in the core of a TU LP game (more specifically, of a flow game) is co-NP-complete.

TU LP games have been studied in a variety of OR applications, including network flows \citep{gui2008dual,markakis2003core} and inventory management \citep{chen2009stoch,chen2016duality,toriello2014dynamic,toriello2017dynamic}.
They are also applicable when the LP game approximates a more complex game with discrete or non-convex elements, such as in facility location \citep{goemans2004cooperative} or vehicle routing \citep{toriello2013uhan}.

\citet{aumann1960von} introduced non-transferable utility (NTU) games as a generalization of TU games in which players may differ in their outcome valuations.
NTU theory is greatly complicated by the fact that the outcomes of a game are no longer measured by real numbers, but by subsets of $\mathbb{R}^N$.
Nevertheless, \citet{scarf1967core}, \citet{billera1970some}, and \citet{shapley1973balanced} showed that balancedness conditions analogous to those of \citet{bondareva1963some} and \citet{shapley1967balanced} for the case of TU games are sufficient, though not necessary \citep{billera1970some}, for the non-emptiness of the core of an NTU game.
Unfortunately, unlike the duality-based results for TU games, these results are not algorithmic and rely on topological arguments.

\section{Non-Transferable Utility Linear Production Games}
\label{sec: planning and cooperation}

In Section~\ref{sec: mathematical framework} we develop the framework of NTU LP games.
In Section~\ref{sec: structural results} we derive structural results regarding the existence, representability, and complexity of the main concept of equilibrium.
In Section~\ref{sec: algorithmic techniques} we develop algorithmic techniques based on mathematical programming.
We refer the reader to Table~\ref{table: notation} for a summary of the notation used throughout this section.

\begin{table}
\centering
\begin{tabular}{@{}l@{\quad}l@{}}
\hline
Symbol & Description                           \\ \hline
$N$    & Set of players, indexed by $i$ \\ 
$J$    & Set of public goods, indexed by $j$ \\ 
$K$    & Set of resources, indexed by $k$ \\ 
$A \in \Q^{K \times J}$    & Production matrix, where $a_{kj}$ is the amount of $k$ needed per unit of $j$ \\ 
$b^i \in \Q^K$  & Resource endowment of $i$ \\ 
$v^i \in \Q^J$  & Good valuation of $i$ \\ 
$X(S) \subseteq \R^J$  & Design space of $\emptyset \neq S \subseteq N$, refer to \eqref{eq: X(S)} \\ 
$Z(S) \subseteq \R^J \times \R^N$  & Design-utility space of $\emptyset \neq S \subseteq N$, refer to \eqref{eq: Z(S)} \\ 
$U(S) \subseteq \R^N$  & Utility space of $\emptyset \neq S \subseteq N$, refer to \eqref{eq: U(S)}  \\ 
$C(N) \subseteq \R^N$  & Core of the game, refer to \eqref{eq: C(N)}  \\ \hline
\end{tabular}
\caption{Summary of notation.}
\label{table: notation}
\end{table}

\subsection{Mathematical Framework}
\label{sec: mathematical framework}

Let $N$, $K$ and $J$ respectively be a finite set of players, resources, and public goods.
Let $A \in \Q^{K \times J}$ be a \emph{production matrix}, where $a_{kj}$ is the amount of resource $k \in K$ needed per unit of good $j \in J$.
Each player $i \in N$ has a fixed resource endowment vector $b^i \in \Q^K$.

The role of the public decision-maker is to administer the collective societal resources $b(N) = \sum_{i \in N} b^i$ to advance their strategic planning goals.
Formally, their design space is given by 
\begin{equation}
\label{eq: X(N)}
    X(N) \coloneqq \{x \in \R_{\geq 0}^J : Ax \leq b(N) \}.
\end{equation}
In turn, they are interested in implementing some design $x^* \in \argmax_{x \in X(N)} g(x)$, where $g: \R^J \rightarrow \R$ is some objective function that reflects their strategic planning goals.
We further assume the players have linear valuations for the goods produced.
That is, each player $i \in N$ has a fixed good valuation vector $v^i \in \Q^J$, so that their utility for a design $x \in \Q^J$ is given by
\begin{equation}
\label{eq: u_i}
    u_i: x \mapsto \left(v^i\right)^T x.
\end{equation}
In light of \eqref{eq: u_i}, the objective function $g$ might capture one of the various conceptualizations of fairness and social welfare~\citep{chen2023guide}.

However, the prospective implementation of $x^*$ need not be well-received by all players.
Suppose there exist some \emph{coalition} $\emptyset \neq S \subseteq N$ with collective resources $b(S) = \sum_{i \in S} b^i$ and some design
\begin{equation}
\label{eq: X(S)}
    x^S 
    \in X(S) 
    \coloneqq \{x \in \R_{\geq 0}^J : Ax \leq b(S) \},
\end{equation}
available to the members of $S$ through their unilateral cooperation, such that
\begin{equation}
\label{eq: blocking}
    u_i^S = (v^i)^T x^S > (v^i)^T x^* = u_i^*, \quad \forall i \in S.
\end{equation}
Then, the members of $S$ would have an incentive to \emph{block} the implementation of $x^*$ in favor of their  implementation of $x^S$.
An effective public decision-maker must optimize for their strategic planning goals while averting this kind of implementation failure.

The \emph{core} of an NTU LP game is the set of utility allocations for which blocking in the sense of \eqref{eq: blocking} is not possible.
Formally, for any coalition $\emptyset \neq S \subseteq N$, the utility maps \eqref{eq: u_i} for players $i \in S$ transform their design space \eqref{eq: X(S)} into their design-utility space
\begin{equation}
\label{eq: Z(S)}
    Z(S) 
    \coloneqq 
    \left\{ (x, u) \in \R^J \times \R^N : 
    \begin{array}{rlc}
        x &\in X(S) &  \\
        u_i &\leq {\left(v^i\right)}^Tx, & \forall i \in S
    \end{array}
    \right\},
\end{equation}
where the inequalities ensure the set of utilities is downward closed; a standard assumption in the NTU literature \citep[Definition 11.3.1]{peleg2007introduction}.
In this way, players are assumed to be free to dispose of their utility. 
We recover their utility space as 
\begin{equation}
\label{eq: U(S)}
    U(S) 
    \coloneqq \textrm{proj}_u \left( Z(S) \right),
\end{equation}
and their design space as $X(S) = \textrm{proj}_x \left( Z(S) \right)$.
Note that \eqref{eq: Z(S)} and \eqref{eq: U(S)} are cylindrical along the coordinates $i \in N \setminus S$: if $u \in U(S)$ and $u' \in \mathbb{R}^N$ satisfies $u_i = u_i'$ for all $i \in S$, then $u' \in U(S)$.
The collection $\{U(S)\}_{\emptyset \neq S \subseteq N}$ of utility spaces given by \eqref{eq: U(S)}, which are compactly represented by $A$, $\{b^i\}_{i \in N}$, and $\{v^i\}_{i \in N}$, define the NTU LP game.
Its core is defined as follows, where $\interior{\cdot}$ is the interior operator.
\begin{definition}
\label{def: core}
The core of the NTU LP game $\{U(S)\}_{\emptyset \neq S \subseteq N}$ given by \eqref{eq: U(S)} is 
\begin{equation}
\label{eq: C(N)}
    C(N) \coloneqq U(N) \setminus \bigcup_{\emptyset \neq S \subseteq N} \interior{U(S)}.
\end{equation}
\end{definition}
In other words, since $\interior{U(S)}$ is downward closed, it is interpreted as the set of utility allocations coalition $\emptyset \neq S \subseteq N$ could object to based on the availability of a strongly dominating alternative.
Therefore, $C(N)$ is the set of utility allocations available to the grand coalition $N$ that no coalition $\emptyset \neq S \subseteq N$ could object to.
An effective public decision-maker proposes the implementation of some design $x^* \in X(N)$ that optimizes their strategic planning goals, embedded in the choice of objective function $g$, \emph{subject to the requirement that its corresponding utility allocation satisfies $u^* \in C(N)$.}
Note that $g$ does not itself shape the set $C(N)$ of core-stable solutions; it instead shapes the subset of these that align most closely with the decision-maker's goals. 
This motivates two fundamental questions: is this possible; and if so, can it be achieved efficiently?

The remainder of this work deals with these questions relative to Definition~\ref{def: core}.
The results in Section~\ref{sec: structural results} are concerned with the structure of $C(N)$ itself and are therefore agnostic to the choice of objective function $g$.
Conversely, in Section~\ref{sec: algorithmic techniques}, we develop a cutting plane algorithm to optimize linear objective functions $g$ over extended formulations of the core; this in particular includes the optimization of various non-linear social welfare functions subject to core membership.

Before continuing, we note that while defining $C(N)$ based on the absence of strongly dominating objections (i.e., as in \eqref{eq: C(N)}) is standard in the NTU literature \citep[Chapter 12]{peleg2007introduction}, other variants of core-like stability are possible.
For example, one may impose cooperation structures that, due to exogenous factors (e.g., sociological, geographical, regulatory), restrict the set of coalitions that can negotiate \citep{myerson1977graphs,myerson1980conference}; this amounts to restricting the coalitions $\emptyset \neq S \subseteq N$ appearing in the union term in \eqref{eq: C(N)}.
Similarly, one may capture additional difficulties or penalties associated with seceding from the grand coalition by requiring coalitions attain a certain level of objection before becoming blocking coalitions~\citep{shapley1966quasi, schulz2010sharing}; in \eqref{eq: C(N)}, this amounts to replacing $\interior{U(S)}$ with the smaller set obtained by subtracting a given $\epsilon > 0$ coordinate-wise from the utility outcomes available to each coalition $\emptyset \neq S \subsetneq N$.
Conversely, one may capture additional benefits associated with remaining in the grand coalition, and in particular exogenous subsidies that encourage cooperation, by augmenting the grand coalition's budget $b(N)$ in \eqref{eq: X(N)} through a multiplicative factor or additive term representing proportional or fixed subsidies, respectively.
In short, the wealth of core-like variants prevalent in the TU literature may be similarly adopted in our setting by adjusting \eqref{eq: Z(S)} or \eqref{eq: C(N)} as necessary.

\subsection{Structural Results}
\label{sec: structural results}

We assume throughout that the sets $X(\{i\})$ are non-empty and bounded for all players $i \in N$.
In turn, this implies that the sets $X(S)$ are non-empty and bounded for all coalitions $\emptyset \neq S \subseteq N$.
For example, the non-emptiness assumptions hold if $b^i \in \Q_{\geq 0}^{K}$ for all $i \in N$.

\subsubsection{Non-Emptiness}
\label{sec: non-emptiness}

Our first result is a set of sufficient conditions for $C(N) \neq \emptyset$ in the special case of NTU LP games.
\citet{scarf1967core} established sufficient conditions (though not necessary, see \citet{billera1970some}) for the core of a (general) NTU game to be non-empty.
Scarf's conditions build on the notion of \emph{balanced} games, which were originally studied by \citet{bondareva1963some} and \citet{shapley1967balanced} in the context of TU games.
A collection of coalitions $\mathcal{S} \subseteq 2^N \setminus \{\emptyset\}$ is balanced if there exist non-negative weights $\lambda^S \geq 0$ for $S \in \mathcal{S}$ such that $\sum_{S \in \mathcal{S}: i \in S} \lambda^S = 1$ for all $i \in N$.
Note that the set of balanced collections generalizes the set of set partitions.
Scarf termed a (general) NTU game $\{U(S)\}_{\emptyset \neq S \subseteq N}$ as balanced if
\begin{equation}
\label{eq: balanced}
    \bigcap_{S \in \mathcal{S}} U(S) \subseteq U(N)
\end{equation}
for any balanced collection $\mathcal{S} \subseteq 2^N \setminus \{\emptyset\}$.
Since the sets $\{U(S)\}_{\emptyset \neq S \subseteq N}$ are cylindrical, if $\mathcal{S}$ is a set partition, then \eqref{eq: balanced} requires that any utility outcome independently available to a coalition $S \in \mathcal{S}$ is also available to $S$ if they join the grand coalition $N$, i.e., $S$ can only benefit from joining $N$.
The following holds if \eqref{eq: balanced} is more generally required for all balanced collections.
\begin{theorem}[\citeauthor{scarf1967core},\citeyear{scarf1967core}]
\label{theorem: scarf}
If $\{U(S)\}_{\emptyset \neq S \subseteq N}$ is a balanced NTU game, then $C(N) \neq \emptyset$.
\end{theorem}

In light of Theorem~\ref{theorem: scarf}, we establish sufficient conditions for an NTU LP game to be balanced.
In fact, our result generalizes a class of NTU games that \citet[Section 2]{scarf1967core} showed to be balanced, namely those arising from exchange economies.

For any coalition $\emptyset \neq S \subseteq N$, the dual cone of their design space $X(S)$ is given by
\begin{equation*}
    X^*(S) \coloneqq \{v \in \R^J: v^T x \geq 0, \forall x \in X(S)\}.
\end{equation*}
We now use these cones to derive conditions under which NTU LP games are balanced.
We first use the weights of a balanced collection to show that any designs independently available to its coalitions can be combined into a design available to the grand coalition.
We then show that if all valuation vectors lie in these cones, then the players can only benefit from the combined design.
\begin{theorem}
\label{theorem: balanced}
Let $\{U(S)\}_{\emptyset \neq S \subseteq N}$ be the NTU LP game given by $A$, $\{b^i\}_{i \in N}$, and $\{v^i\}_{i \in N}$.
If 
\begin{equation*}
    v^i \in \bigcap_{S \in 2^N \setminus \{\emptyset\}} X^*(S) 
\end{equation*}
for all $i \in N$, the game is balanced and satisfies $C(N) \neq \emptyset$.
\end{theorem}
\begin{proof}{Proof}
Given Theorem~\ref{theorem: scarf}, it remains to show that \eqref{eq: balanced} holds for an arbitrary balanced collection $\mathcal{S} \subseteq 2^N \setminus \emptyset$ with balancing weights $\lambda^S \geq 0$ for $S \in \mathcal{S}$.
If $\bigcap_{S \in \mathcal{S}} U(S) = \emptyset$, then the claim follows immediately.
Therefore, suppose $\bigcap_{S \in \mathcal{S}} U(S) \neq \emptyset$ and consider any $u \in \bigcap_{S \in \mathcal{S}} U(S)$.
We need to show that $u \in U(N)$.
By assumption, for each $S \in \mathcal{S}$ there exists $x^S \in \R^{J}$ with (\textit{i}) $x^S \in \R_{\geq 0}^{J}$, (\textit{ii}) $A x^S \leq b(S)$, and (\textit{iii}) $u_i \leq \left(v^i\right)^T x^S$ for each $i \in S$.

Now, let $x = \sum_{S \in \mathcal{S}} \lambda^S x^S$.
First, note that (\textit{i}) implies (\textit{iv}) $x \in \R_{\geq 0}^J$.
Moreover, 
\begin{align*}
    Ax 
    & = A\left(\sum_{S \in \mathcal{S}} \lambda^Sx^S\right) 
    = \sum_{S \in \mathcal{S}} \lambda^SAx^S 
    \stackrel{(1)}{\leq} \sum_{S \in \mathcal{S}} \lambda^Sb(S) \\
    & = \sum_{S \in \mathcal{S}} \lambda^S \left(\sum_{i \in S} b^i \right)
    = \sum_{i \in N} b^i \sum_{S \in \mathcal{S}: i \in S} \lambda^S 
    \stackrel{(2)}{=} \sum_{i \in N} b^i \\
    & = b(N),
\end{align*}
where $(1)$ follows from (\textit{ii}) and $(2)$ holds since $\mathcal{S}$ is a balanced collection of $N$ with balancing weights $\lambda^S \geq 0$ for $S \in \mathcal{S}$.
This shows that (\textit{v}) $A x \leq b(N)$.
Next, consider any $i \in N$.
Note that
\begin{align*}
    \left(v^i\right)^T x 
    &= \left(v^i\right)^T \left( \sum_{S \in \mathcal{S}} \lambda^S x^S \right)
    = \sum_{S \in \mathcal{S}} \lambda^S \left(v^i\right)^T x^S \\
    &= \sum_{S \in \mathcal{S}: i \in S} \lambda^S \left(v^i\right)^T x^S + \sum_{S \in \mathcal{S}: i \notin S} \lambda^S \left(v^i\right)^T x^S \\
    &\stackrel{(1)}{\geq} \sum_{S \in \mathcal{S}: i \in S} \lambda^S \left(v^i\right)^T x^S 
    \stackrel{(2)}{\geq} \sum_{S \in \mathcal{S}: i \in S} \lambda^S u_i \\
    &= u_i \sum_{S \in \mathcal{S}: i \in S} \lambda^S 
    \stackrel{(3)}{=} u_i,
\end{align*}
where $(1)$ holds since $\lambda^S \geq 0$ and $v^i \in X^*(S)$ for all $S \in \mathcal{S}$, $(2)$ follows from (\textit{iii}), and $(3)$ holds since $\mathcal{S}$ is a balanced collection of $N$ with balancing weights $\lambda^S \geq 0$ for $S \in \mathcal{S}$.
This shows that (\textit{vi}) $u_i \leq \left(v^i\right)^T x$ for each $i \in N$.
The choice of $x$ together with (\textit{iv})-(\textit{vi}) imply that $u \in U(N)$.
\end{proof}
The same arguments show that an NTU LP game satisfying the conditions of Theorem~\ref{theorem: balanced} is totally balanced (i.e., all of its sub-games are balanced).
Also, the proof extends to the case in which the design spaces are non-empty after intersecting them with a shared convex set that does not depend on the resource pooling.
As a more interpretable corollary, we find that if the players are at worst indifferent about the various goods, then cooperation is possible.
\begin{corollary}
\label{corollary: balanced}
Let $\{U(S)\}_{\emptyset \neq S \subseteq N}$ be the NTU LP game given by $A$, $\{b^i\}_{i \in N}$, and $\{v^i\}_{i \in N}$.
If $v^i \in \Q_{\geq 0}^J$ for all $i \in N$, then the game is balanced and satisfies $C(N) \neq \emptyset$.
\end{corollary}
\begin{proof}{Proof}
Note that $\R_{\geq 0}^J \subseteq X^*(S)$ for all $S \in 2^N \setminus \{\emptyset\}$.
Therefore, for any $i \in N$, $v^i \in \Q_{\geq 0}^J$ implies $v^i \in X^*(S)$ for all $S \in 2^N \setminus \{\emptyset\}$.
Then, the claim follows from Theorem~\ref{theorem: balanced}.
\end{proof}
Conversely, if the players have an aversion to certain goods (i.e., negative entries in their valuation vectors), the non-emptiness can no longer be so straightforwardly guaranteed.
The following example illustrates how, without the non-negativity conditions, simpler notions of ``sufficiently well-aligned valuation vectors'', such as a being nearly parallel to each other, may be insufficient.
\begin{example}
\label{ex: insufficient}
Consider the NTU LP game with $N = \{1, 2, 3\}$, $J = \{1, 2\}$, and $K = \{1\}$ given by $A = \begin{pmatrix} 1 & 1 \end{pmatrix}$, $v^1 = v^2 = \begin{pmatrix} 2/3 & 1/3 \end{pmatrix}^T$, $v^3 = \begin{pmatrix} -2/3 & 1/3 \end{pmatrix}^T$, and $b^1 = b^2 = b^3 = 1$.
The coalition $S = \{1, 2\}$ can achieve utilities $u_1 = u_2 = 4/3$ by investing their two units of budget on $x_1$.
Therefore, in a core allocation, at least one of players $1$ or $2$ must meet or exceed this level of utility.
This implies that, in a core allocation, $u_3 \leq 0$ (this can be achieved by investing one unit of budget on $x_1$ and two units of budget on $x_2$).
However, player $3$ can achieve utility $u_3 = 1/3$ on their own by investing their unit of budget on $x_2$.
It follows that $C(N) = \emptyset$.
Although $(v^1)^T v^3 = (v^2)^T v^3 = -3/9$, these inner products can be increased arbitrarily without changing the structure of the game by introducing distinguishable copies of good $2$.
Therefore, on their own, assumptions involving the magnitude of $\left(v^i\right)^{T} v^{i'}$ for all $i, i' \in N$ need not guarantee the non-emptiness of the core.
\end{example}

While Theorem~\ref{theorem: balanced} indicates that NTU LP games are synergistic under mild conditions, attaining this synergy is not always clear. 
As we illustrate, mechanisms that only combine the preferred solutions of different players (as in the proof of Theorem~\ref{theorem: balanced}) need not achieve core stability.
\begin{example}
\label{example: synergy}
Consider the NTU LP game with $N = \{1, 2\}$, $J = \{1, 2, 3\}$, and $K = \{1\}$ given by $A = \begin{pmatrix} 1 & 1 & 1\end{pmatrix}$, $v^1 = \begin{pmatrix} 1 & 0 & 2/3 \end{pmatrix}^T$, $v^2 = \begin{pmatrix} 0 & 1 & 2/3 \end{pmatrix}^T$, and $b^1 = b^2 = 1$.
If acting alone, player $1$ finds the design $x = \begin{pmatrix} 1 & 0 & 0 \end{pmatrix}^T$ to be optimal with utilities $u = \begin{pmatrix} 1 & 0 \end{pmatrix}^T$.
If acting alone, player $2$ finds the design $x = \begin{pmatrix} 0 & 1 & 0 \end{pmatrix}^T$ to be optimal with utilities $u = \begin{pmatrix} 0 & 1 \end{pmatrix}^T$.
However, the combined design $x = \begin{pmatrix} 1 & 1 & 0 \end{pmatrix}^T$ with $u = \begin{pmatrix} 1 & 1 \end{pmatrix}^T$ is not in the core, as it is dominated by the design $x = \begin{pmatrix} 0 & 0 & 2 \end{pmatrix}^T$ with $u = \begin{pmatrix} 4/3 & 4/3 \end{pmatrix}^T$.
\end{example}
The remainder of Section~\ref{sec: planning and cooperation} deals with computational and algorithmic aspects related core stability.

\subsubsection{MIP-representability}
\label{sec: mip-representability}

We now show that, if non-empty, the core of an NTU LP game can be represented using mixed-integer linear programming (MIP).
The following standard technical lemma, which we prove in the Appendix, states that the difference between a polyhedron $P$ and the interior of polyhedron $Q$ can be written as a finite union of polyhedra: the union, over the inequalities defining $Q$, of $P$ intersected with each reverse inequality.
\begin{lemma}
\label{lemma: difference}
Let $P \subseteq \R^n$ and $Q = \{x \in \mathbb{R}^n: Ax \leq b\} \subseteq \R^n$ be polyhedra, where $A \in \Q^{m \times n}$ and $b \in \Q^m$.
Then, $P \setminus \interior{Q} = \bigcup_{i = 1}^m P \cap \{x \in \mathbb{R}^n : a_i^Tx \geq b_i\}$.
\end{lemma}
A repeated application of Lemma~\ref{lemma: difference} shows that the core of an NTU LP game can be written as a finite union of polytopes, and is thus representable as an extended MIP formulation using disjunctions.
\begin{theorem}
\label{theorem: mip-representable}
Let $\{U(S)\}_{\emptyset \neq S \subseteq N}$ be the NTU LP game given by $A$, $\{b^i\}_{i \in N}$, and $\{v^i\}_{i \in N}$, and suppose $C(N) \neq \emptyset$.
Then, $C(N)$ is MIP-representable.
\end{theorem}
\begin{proof}{Proof}
Consider the set $U'(N) = U(N) \setminus \bigcup_{i \in N} \interior{U(\{i\})}$.
Note that $U'(N)$ is a polytope and $C(N) = U'(N) \setminus \bigcup_{\emptyset \neq S \subseteq N: |S| > 1} \interior{U(S)}$.
Now, consider any ordering $S_1, S_2, \ldots, S_{2^n - (n + 1)}$ of the sets $\emptyset \neq S \subseteq N$ with $|S| > 1$.
Since $U'(N)$ is a polytope, Lemma~\ref{lemma: difference} implies that $U'(N) \setminus U(S_1)$ is the union of a finite family of polytopes.
In turn, for any polytope $P$ in this family, Lemma~\ref{lemma: difference} implies that $P \setminus \interior{U(S_2)}$ is the union of a finite family of polytopes, so that $\left(U'(N) \setminus \interior{U(S_1)} \right)\setminus \interior{U(S_2)}$ itself is the union of a finite family of polytopes.
Upon repeating the application of this argument, we obtain that $C(N)$ itself is the union of a finite family of polytopes.
A finite family of polytopes has a common (trivial) recession cone, so its union is MIP-representable~\citep{jeroslow1984modelling}; refer also to \citep{vielma2015mixed}.
\end{proof}

\subsubsection{Complexity}
\label{sec: complexity}

To prove Theorem~\ref{theorem: scarf}, \citet{scarf1967core} developed a path-following procedure based on that of \citet{lemke1964equilibrium} for computing exact mixed Nash equilibria of bimatrix games.
In an alternative proof of Theorem~\ref{theorem: scarf}, \citet{shapley1973balanced} developed a path-following argument that generalizes that of \citet{sperner1928neuer} for the existence of colorful triangles in valid colorings of triangulations of the $n$-simplex.
However, while these proofs are constructive in spirit, neither yields an efficient algorithm for finding a point in the core of a balanced NTU game\textemdash they both involve limit arguments.
In other words, for all practical purposes, Theorem~\ref{theorem: scarf} is existential.

\citet{papadimitriou1994complexity} defined the class PPAD (for polynomial parity argument in a digraph) as the set of search problems with a non-empty solution set, a polynomial-time verifier of membership in the solution set, and for which the non-emptiness of the solution set is proved using a directed path-following argument (on a digraph whose size can be exponential in that of the problem input).
For example, the problem of computing approximate mixed Nash equilibria of finite games is in PPAD (this follows from the use of the path-following argument of \citet{sperner1928neuer} in the proof of \citet{nash1950equilibrium}), as is the problem of computing exact mixed Nash equilibria of bimatrix games (this follows from the path-following procedure of \citet{lemke1964equilibrium}); the latter is in fact PPAD-complete~\citep{chen2009settling,daskalakis2009complexity}.

Based on the use of path-following techniques in the proofs of Theorem~\ref{theorem: scarf}, it is natural to ask about the membership and completeness in PPAD of the problem of finding a point in the core of a balanced (general) NTU game.
Indeed, \citet{kintali2013reducibility} showed that the problem is PPAD-complete when the game has an exponential-sized representation (i.e., given by an explicit list of possible coalitions and Pareto-optimal outcomes).
We show a stronger negative result for the case for balanced NTU LP games with a compact representation given by $A$, $\{b^i\}_{i \in N}$, and $\{v^i\}_{i \in N}$.
Specifically, while the non-emptiness of the core is guaranteed by Theorem~\ref{theorem: scarf}, there is no polynomial-time verifier for membership in the core unless $\text{P} = \text{co-NP}$.
As we show, this result holds even for highly-structured instances.

\begin{theorem}
\label{theorem: co-np-complete}
Let $\{U(S)\}_{\emptyset \neq S \subseteq N}$ be a balanced NTU LP game given by $A$, $\{b^i\}_{i \in N}$, and $\{v^i\}_{i \in N}$.
The problem of deciding whether a given $u^* \in U(N)$ satisfies $u^* \in C(N)$ is co-NP-complete, even if $|K| = 1$, $A$ is a row matrix of ones, and $b^i = 1$, $v^i \in \mathbb{Q}_{\geq 0}^J$, and $u_i^* = 1$ for all $i \in N$.
\end{theorem}
\begin{proof}{Proof}
To show that the problem is in co-NP, consider a pair $(S, u^S)$ as a ``No'' certificate for the problem instance, where $\emptyset \neq S \subseteq N$ and $u^S \in U(S)$.
In time polynomial in $A$, $\{b^i\}_{i \in N}$, and $\{v^i\}_{i \in N}$, one can decide whether $u_i^S > u_i^*$ for all $i \in S$.
In this case, $S$ is indeed a blocking coalition against $u^*$, and the answer to the problem instance is ``No.''

To show that the problem is co-NP-hard, we equivalently show that its complement (the problem of deciding whether $u^* \notin C(N)$) is NP-hard.
We do this through a reduction from the three-dimensional perfect matching problem (\texttt{3DM}).
In \texttt{3DM}, we are given three disjoint sets $X$, $Y$, and $Z$, each of size $n \in \mathbb{N}$.
We are also given a collection of triples $T \subseteq X \times Y \times Z$, of size $m \in \mathbb{N}$.
The elements of $H = X \cup Y \cup Z$ are nodes and the elements of $T$ are (hyper)edges, so that $(H, T)$ forms a hypergraph.
The problem is to decide whether there exists a subcollection $T' \subseteq T$ of $n$ edges such that each node $h \in H$ is contained in exactly one edge $t \in T'$.
That is, to decide whether $(H,T)$ contains a three-dimensional perfect matching.
\citet{karp1972reducibility} showed that \texttt{3DM} is NP-complete.

Given an arbitrary instance $\left(H, T\right)$ of \texttt{3DM} with $n > 1$ and $m \geq n$, first duplicate an arbitrary edge until $\frac{1}{4n} \left(1 - \frac{1}{2(n-1)(4n-1)} \right) > \frac{1}{3n + m}$.
Note that this does not change its decision property.
Now, consider the following NTU LP game.
Let $K = \{1\}$, $J = \{1, 2, \ldots, m, m + 1\}$, and $A = \begin{pmatrix}1 & \cdots & 1 \end{pmatrix}$ be the row matrix of ones of appropriate dimension.
In this way, the first $m$ goods are ``edge'' goods.
With some abuse of notation, identify the index $j \in J$ of an edge good with its corresponding edge $t \in T$.
For each node $h \in H$, define a ``node'' player $h$ with $b^h = 1$ and $v^h \in \mathbb{Q}_{\geq 0}^J$ given by
\begin{equation*}
    v_j^h 
    = 
    \begin{cases}
        \frac{1}{4n - 1}, & \text{ if } j \in T \land h \in j \\ 
        \frac{1}{4n} \left(1 - \frac{1}{2(n-1)(4n-1)} \right), & \text{ if } j \in T \land h \notin j \\
        \frac{1}{3n + m}, & \text{ if } j = m + 1
    \end{cases}.
\end{equation*}
Similarly, for each edge $t \in T$, define an ``edge'' player $t$ with $b^t = 1$ and $v^t \in \mathbb{Q}_{\geq 0}^J$ given by
\begin{equation*}
    v_j^t 
    = 
    \begin{cases}
        \frac{n}{4n - 1}, & \text{ if } j \in  T \land j = t \\ 
        0, & \text{ if } j \in T \land j \neq t \\
        \frac{1}{3n + m}, & \text{ if } j = m + 1
    \end{cases}.
\end{equation*}
In this way, $N = H \cup T$ and $|N| = 3n + m$.
This specifies the NTU LP game $\{U(S)\}_{\emptyset \neq S \subseteq N}$ given by $A$, $\{b^i\}_{i \in N}$, and $\{v^i\}_{i \in N}$.
By Corollary~\ref{corollary: balanced}, the game is balanced and satisfies $C(N) \neq \emptyset$.
Finally, let $u_i^* = 1$ for all $i \in N$.
Note that $u^* \in U(N)$, as it can be achieved by allocating all of $b(N) = 3n + m$ to (the production of) good $m + 1 \in J$.
This completes the description of the problem instance.

Deciding whether $u^* \notin C(N)$ is equivalent to deciding whether there exists a blocking coalition $\emptyset \neq S \subseteq H \cup T$ against $u^*$.
Therefore, we show that such a blocking coalition exists if and only if $(H,T)$ contains a three-dimensional perfect matching.
We first derive some structural properties of blocking coalitions for this problem instance; their proofs can be found in the Appendix.
\begin{lemma}
\label{lemma: hardness structural lemma 1}
If $S \subseteq H \cup T$ is a blocking coalition against $u^*$, then $|S \cap H| = 3n$ and $|S \cap T| = n$.
\end{lemma}
\begin{lemma}
\label{lemma: hardness structural lemma 2}
If $S \subseteq H \cup T$ is a blocking coalition against $u^*$, then for every $h \in S \cap H$ there exists some $t \in S \cap T$ such that $h \in t$.
\end{lemma}

Now, suppose there exists a blocking coalition $S \subseteq H \cup T$ against $u^*$.
Lemma~\ref{lemma: hardness structural lemma 1} implies that $|S \cap H| = 3n$, so that $H \subseteq S$.
Together with Lemma~\ref{lemma: hardness structural lemma 2} this implies that, for every node $h \in H$, there exists some edge $t \in S \cap T$ such that $h \in t$.
Lastly, since $T \subseteq X \times Y \times Z$ and Lemma~\ref{lemma: hardness structural lemma 1} implies that $|S \cap T| = n$, it must be that $S \cap T$ partitions $H$.
In other words, $T' = S \cap T$ forms a three-dimensional perfect matching of $(H,T)$.

Conversely, suppose $T' \subseteq T$ forms a three-dimensional perfect matching of $(H,T)$ and let $S = T' \cup H$, so that $|S \cap H| = 3n$ and $|S \cap T| = n$.
Note that $b(S) = |S| = 4n$ and consider the uniform allocation of $b(S)$ to the edge goods $S \cap T$, so that each good $t \in S \cap T$ is allocated $4$ units of budget.
Then, the utility of each edge player $t \in S \cap T$ is given by 
\begin{equation*}
    4 \cdot \underbrace{\frac{n}{4n-1}}_{v_t^t} = 1 + \frac{1}{4n-1} > 1.
\end{equation*}
Similarly, the utility of each node player $h \in S \cap H$ is given by
\begin{align*}
    4 \cdot \underbrace{\frac{1}{4n-1}}_{v_t^h \text{ if } h \in t} + & (4n - 4) \cdot \underbrace{\left(\frac{1}{4n} \left(1 - \frac{1}{2(n-1)(4n-1)} \right)\right)}_{v_t^h \text{ if } h \notin t} \\
    &= 1 + \frac{1}{2n(4n-1)} 
    > 1.
\end{align*}
Therefore, $S$ forms a blocking coalition against $u^*$.
This completes the proof.
\end{proof}

Note that the theorem does not preclude the possibility that verifying membership in the core of a balanced NTU LP game can be performed in polynomial time for \emph{some} of its members.
As an analogy, in the case of TU LP games, this is the distinction between the points in the core that can be obtained through the methods of \citet{owen1975core} and those that arise in the $\text{co-NP}$-completeness proof of \citet{fang2001membership} for the case of flow games.

However, we conjecture there is a stronger version of Theorem~\ref{theorem: co-np-complete} that holds for all points in the core of an NTU LP game.
In our proof we consider the property $u^* \in C(N)$, where $u^*$ is an all ones vector obtained through a highly-coordinated strategy in which the entire budget available to the grand coalition is dedicated to the production of a single good that is only mildly preferred by each player.
Alternatively, one could modify our gadget to obtain the same $u^*$ through a completely uncoordinated strategy, by introducing additional goods whose player-valuation structure is an identity matrix.
In this way, each player could guarantee themselves unit utility, with the hope to ultimately have that if $u^* \in C(N)$, then $C(N) = \{u^*\}$.
Unfortunately, this modification complicates the combinatorial arguments needed for structural lemmas analogous to Lemmas~\ref{lemma: hardness structural lemma 1} and \ref{lemma: hardness structural lemma 2}.

Lastly, we note some practical interpretations of Theorem~\ref{theorem: co-np-complete}.
It suggests that, in general, it is difficult for public decision-makers to recognize whether a proposed solution will be supported by \emph{all} population groups.
Equally, in the worst case, it is difficult for constituents to recognize whether they can organize to improve upon a proposed solution.
This does not discredit the plausibility of constituents forming blocking coalitions under most circumstances.
Instead, it points at the fundamental difficulty of \emph{certifying} popular support: it limits the extent to which public decision-making can be expected to be robust against (possibly algorithmically concerted) coalition formation and thus informs how legal and/or regulatory frameworks may align with its mathematical nature.
This parallels the mathematical, political, and regulatory interplay of adjacent problem areas such as political districting \citep{ellenberg2021geometry} and spectrum allocation~\citep{leyton2017economics}.
Finally, the theorem suggests that this difficulty is inherent to the presence of disparate stakeholder interests.

\subsection{Algorithmic Techniques}
\label{sec: algorithmic techniques}

Theorem~\ref{theorem: mip-representable} implies that, in principle, one can write an extended MIP formulation of the core of an NTU LP game using disjunctive constraints, one for each $\emptyset \neq S \subseteq N$.
However, this approach is unlikely to yield efficient algorithms, not only because of the number of disjunctions, but also because of their individual size.
As we illustrate next, the facial structure of $\{U(S)\}_{\emptyset \neq S \subseteq N}$ may require a very large number of linear inequalities, even for highly-structured and balanced instances.
\begin{example}
\label{ex: cyclic}
Let $K = \{1\}$, $J = \{1, 2, \ldots, m\}$, and $N = \{1, 2, \ldots, n\}$ for some $n, m \in \mathbb{N}$.
Let $A = \begin{pmatrix}1 & \cdots & 1 \end{pmatrix}$ be the row matrix of ones of appropriate dimension, and take any real numbers $1 < t_1 < t_2 < \cdots < t_m$.
For each $i \in N$, let $b^i = 1$ and $v_j^i = t_j^i$ for $j \in J$.
Then, for any $\emptyset \neq S \subseteq N$, $X(S)$ is the $|S|$-dilation of a standard simplex.
For each $j \in J$, its extreme point $e_j \cdot |S| \in X(S)$ (where $e_j$ is a standard vector) maps to utility $|S| \cdot t_j^i$ for $i \in S$.
Therefore, $U(S)$ embeds $\conv{\{(|S| \cdot t_j^i)_{i \in S}\}_{j \in J}}$ along its non-cylindrical coordinates, where $\conv{\cdot}$ denotes the convex hull operator.
This is the \emph{cyclic polytope}, which achieves the largest number of facets for any given dimension and number of vertices~\citep{mcmullen1970maximum}.
\end{example}
This difficulty motivates the development of a cutting plane algorithm.

\subsubsection{Testing Membership}
\label{sec: testing membership}
As a first step, we revisit the problem of testing membership in the core of an NTU LP game.
While this problem is co-NP-complete by Theorem~\ref{theorem: co-np-complete}, we can nevertheless write a MIP formulation:
\begin{subequations}
\label{eq: membership}
\begin{align}
    \epsilon(u^*) \coloneqq \max_{\epsilon, y, x^S, u^S} \epsilon \label{obj: epsilon} \\
    \textrm{ s.t. } \sum_{i \in N} y_i &\geq 1, & \label{cnstr: size} \\
    Ax^S &\leq \sum_{i \in N} b^{i} y_i, & \label{cnstr: x^S} \\
    u_i^S &\leq {\left(v^i\right)}^T x^S & \forall i \in N \label{cnstr: u^S} \\
    \epsilon &\leq u_i^S - u_i^* y_i + M(1 - y_i), & \forall i \in N \label{cnstr: epsilon} \\
    x^S &\geq 0, & \label{cnstr: x^S >= 0} \\
    y &\in \{0,1\}^N . &
\end{align}
\end{subequations}
In \eqref{eq: membership}, $A$, $\{b^i\}_{i \in N}$, $\{v^i\}_{i \in N}$, $u^* \in U(N)$, and a sufficiently big $M > 0$ are all given.
The problem is to decide whether $u^* \in C(N)$, i.e., whether there exist $\emptyset \neq S \subseteq N$ and $u^S \in U(S)$ such that $u_i^S > u_i^*$ for all $i \in S$.
The variables $y_i \in \{0,1\}$ indicate whether $i \in N$ belongs to the candidate blocking coalition against $u^*$, so that $S = \{i \in N: y_i = 1\}$.
Constraint~\eqref{cnstr: size} ensures that $S \neq \emptyset$.
Constraints~\eqref{cnstr: x^S}, \eqref{cnstr: u^S}, and \eqref{cnstr: x^S >= 0} ensure that $(x^S, u^S) \in Z(S)$.
Constraints~\eqref{cnstr: epsilon} ensure that the upper bound on $\epsilon$ is adopted for $i \in S$ and discarded for $i \in N \setminus S$ (by the big $M$).
Then, $\epsilon \leq u_i^S - u_i^*$ for all $i \in S$, so that the objective \eqref{obj: epsilon} is to maximize the least additive utility gain across all members of $S$.
Note that $\epsilon(u^*) \geq 0$ since setting $y_i = 1$ and $u_i^S = u_i^*$ for all $i \in N$ induces a feasible solution with $\epsilon = 0$.

In a commercial solver with a dedicated indicator constraint feature, \eqref{cnstr: epsilon} is often best implemented in its non-linear indicator form $y_i^S \left(\epsilon - \left(u_i^S - u_i^* \right) \right) \leq 0$ for all $i \in N$ for numerical stability.
We refer to $\epsilon(u^*)$ as the (additive) least objection.
Given \eqref{eq: membership}, we immediately obtain the following.
\begin{lemma}
\label{lemma: membership}
Let $\{U(S)\}_{\emptyset \neq S \subseteq N}$ be an NTU LP game given by $A$, $\{b^i\}_{i \in N}$, and $\{v^i\}_{i \in N}$, and let $u^* \in U(N)$.
Then, $u^* \in C(N)$ if and only if $\epsilon(u^*) = 0$.
\end{lemma}

\subsubsection{Cutting Plane Algorithm}
\label{sec: cutting plane algorithm}
With Lemma~\ref{lemma: membership} in place, our goal is to leverage it as a separation procedure.
For technical reasons that we expand on later, it is convenient to do so in the design-utility space $\R^J \times \R^N$ rather than the utility space $\R^N$.
To this end, for each $\emptyset \neq S \subseteq N$, let 
\begin{equation*}
    U'(S) \coloneqq \{ (x, u) \in \R^J \times \R^N : u \in U(S) \}.
\end{equation*}
Note that $U'(S)$ is cylindrical along the coordinates $j \in J$ and $i \in N \setminus S$.
In other words, $U'(S)$ is an extended formulation of $U(S)$.
This implies the following.
\begin{lemma}
\label{lemma: U - U'}
Let $(x, u) \in \R^J \times \R^N$.
For any $\emptyset \neq S \subseteq N$, $u \in \interior{U(S)}$ if and only if $(x, u) \in \interior{U'(S)}$.
\end{lemma}

First, we show an equivalence between the core, as defined in \eqref{eq: C(N)} in the utility space $\R^N$, and an analogous set in the design-utility space $\R^J \times \R^N$.
Let
\begin{align}
\label{eq: C'(N)}
    C'(N) 
    &\coloneqq Z(N) \setminus \bigcup_{\emptyset \neq S \subseteq N}  \interior{U'(S)} \nonumber \\
    &= \bigcap_{\emptyset \neq S \subseteq N}  Z(N) \setminus\interior{U'(S)},
\end{align}
where the equality follows from De Morgan's laws.
We prove the following in the Appendix.
\begin{lemma}
\label{lemma: equivalence}
$C(N) = \projection{u}{C'(N)}$.
\end{lemma} 
In other words, $C'(N)$ is an extended formulation of $C(N)$.

For any $\emptyset \neq S \subseteq N$, the set $Z(N) \setminus \interior{U'(S)}$ in the intersection form of \eqref{eq: C'(N)} is a \emph{reverse convex set}: the points in a given polyhedron that lie outside a given open (in this case polyhedral) convex set.
Therefore, $C'(N)$ can be seen as the intersection of $2^n - 1$ reverse convex sets.

Reverse convex sets appear in various areas of mathematical optimization, including concave minimization, integer programming, and mixed-integer nonlinear programming.
Given a polyhedral set $F$, a closed convex set $O$, and a basic solution $f^* \in F \cap O$, an \emph{intersection cut}~\citep{tuy1964concave,balas1971intersection} is a valid inequality for $\conv{F \setminus \interior{O}}$ that cuts $f^*$.
The cut is generated using information from the extreme rays of a translated simplicial cone, derived from a simplex tableau associated with $f^*$, which contains $F$ and has $f^*$ as its apex.
Specifically, each extreme ray $r \in \mathcal{R}$ corresponds to a non-basic variable $f_r$, and one computes 
\begin{equation}
\label{eq: zeta_r}
    \zeta_r \coloneqq \max \{ \zeta \geq 0:  f^* + \zeta r \in O\}.
\end{equation}
Then, with the convention that $1/\infty = 0$, the inequality $\sum_{r \in \mathcal{R}} \frac{1}{\zeta_r}f_r \geq 1$ is valid for $\conv{F \setminus \interior{O}}$; refer to Figure~\ref{fig: intersection cuts} for a schematic diagram and \citet[Chapter 6]{conforti2014integer} for details.
\begin{figure}
    \centering
    \caption{
        Schematic diagram of intersection cuts.
        In Figure~\ref{fig: rays}, there is a polytope $F$, a closed (in this case polyhedral) convex set $O$, and a vertex $f^* \in F \cap \interior{O}$. 
        The extreme rays of the simplicial cone with apex $f^*$, along with their intersections \eqref{eq: zeta_r} with the boundary of $O$, are shown.
        In Figure~\ref{fig: cut}, the intersections define the cut necessary to obtain $\conv{F \setminus \interior{O}}$.
    }
    \label{fig: intersection cuts}
    \begin{subfigure}{0.45\linewidth}
    \centering
    \resizebox{0.65\linewidth}{!}{
        \begin{tikzpicture}
            \fill[cugold,opacity=1/3] (0,0) -- (2,-1) -- (4,1) -- cycle;
            \draw[->,black,line width=0.5pt,opacity=1] (3,-1.5) -- (-1,0.5) -- (-1 + 0.5*1/2.236,0.5 + 0.5*2/2.236);
            \draw[->,black,line width=0.5pt,opacity=1] (5,2) -- (1,-2) -- (1 - 0.5*1/1.414,-2 + 0.5*1/1.414);
            \draw[->,black,line width=0.5pt,opacity=1] (5,1.25) -- (-1,-0.25) -- (-1 + 0.5*1/4.123,-0.25 - 0.5*4/4.123);

            \fill[reddishpurple,opacity=1/5] (2.5,0.25) -- (5, -0.5) -- (4.5,2.5) -- cycle;
            \draw[black,dashed,line width=0.5pt,opacity=1] (4.5,2.5) -- (2.5,0.25) -- (5, -0.5);

            \fill[black] (4,1) circle (3pt);
            \draw[->,black,line width=1.5pt,opacity=1] (4,1) -- (0 + 0.73*4,0 + 0.73*1);
            \draw[->,black,line width=1.5pt,opacity=1] (4,1) -- (2 + 0.54*2,-1 + 0.54*2);

            \node[black, anchor=center] at (2,-0.25) {$F$};
            \node[black, anchor=center] at (4.25,0.25) {$\textrm{int}(O)$};
            \node[black, anchor=center] at (4.125,1.375) {$f^*$};
            
    \end{tikzpicture}
    }
    \caption{
    $F \setminus \interior{O}$.
    }
    \label{fig: rays}
    \end{subfigure}
    \begin{subfigure}{0.45\linewidth}
    \centering
    \resizebox{0.65\linewidth}{!}{
        \begin{tikzpicture}
            
            \fill[cugold,opacity=1/3] (0,0) -- (2,-1) -- (3.08,0.08) -- (2.92,0.73) -- cycle;
            \draw[->,black,line width=0.5pt,opacity=1] (3,-1.5) -- (-1,0.5) -- (-1 + 0.5*1/2.236,0.5 + 0.5*2/2.236);
            \draw[->,black,line width=0.5pt,opacity=1] (5,2) -- (1,-2) -- (1 - 0.5*1/1.414,-2 + 0.5*1/1.414);
            \draw[->,black,line width=0.5pt,opacity=1] (5,1.25) -- (-1,-0.25) -- (-1 + 0.5*1/4.123,-0.25 - 0.5*4/4.123);


            \draw[->,black,line width=0.5pt,opacity=1] (3.33,-0.94) -- (2.67,1.749) -- (2.67 - 0.5*0.9715,1.749 - 0.5*0.2385);     
    \end{tikzpicture}
    }
    \caption{$\conv{F \setminus \interior{O}}$.}
    \label{fig: cut}
    \end{subfigure}
\end{figure}

To leverage this structure algorithmically, note that in the intersection form of \eqref{eq: C'(N)}, we may replace the common term $Z(N)$ with any polyhedron $P \subseteq \R^J \times \R^N$ satisfying $\conv{C'(N)} \subseteq P' \subseteq Z(N)$ without affecting the left-hand side of the expression nor its structure as the intersection of reverse convex sets.
Therefore, and based on Lemma~\ref{lemma: equivalence}, our technique is to maintain a polyhedral relaxation $P'$, starting with $P' = Z(N)$, that iteratively approaches $\conv{C'(N)}$.
In particular, given any extreme point $(x^*, u^*) \in P'$ (e.g., an optimal solution for a linear objective), the existence of a blocking coalition against $u^*$ implies the possibility of generating an intersection cut.
We prove the following (more general) statement in the Appendix.
\begin{lemma}
\label{lemma: conv(P'-int(U'(S))}
Let $P' \subseteq \R^J \times \R^N$ be a polyhedron.
For any $\emptyset \neq S \subseteq N$, if  $(x^*, u^*) \in P'$ is an extreme point satisfying $(x^*, u^*) \in \interior{U'(S)}$, then $(x^*, u^*) \notin \conv{P' \setminus \interior{U'(S)}}$.
\end{lemma}

Lemma~\ref{lemma: conv(P'-int(U'(S))} does not apply to projected subspaces, however.
In particular, even if $(x^*, u^*) \in P'$ is an extreme point satisfying $u^* \in \interior{U(S)}$, so that $(x^*, u^*) \notin \conv{P' \setminus \interior{U'(S)}}$ by Lemma~\ref{lemma: conv(P'-int(U'(S))}, it may be that $\projection{u_S}{u^*} \in \interior{\projection{u_S}{U'(S)}}$ while simultaneously $\projection{u_S}{u^*} \in \conv{\projection{u_S}{P'} \setminus \interior{\projection{u_S}{U'(S)}}}$, where $u_S = \{u_i : i \in S\}$.
The reason is that $\projection{u_S}{u^*}$ need not be an extreme point of $\projection{u_S}{P'}$.
In practical terms, this is one reason why our cutting plane algorithm must operate in the design-utility space rather than just in the utility space.

Next, we note that given a polyhedron $P' \subseteq \R^J \times \R^N$ satisfying $\conv{C'(N)} \subseteq P' \subseteq Z(N)$ and an extreme point $(x^*, u^*) \in P'$, any intersection cut derived from Lemma~\ref{lemma: conv(P'-int(U'(S))} will strengthen $P'$ as a relaxation of $\conv{C'(N)}$.
We prove this in the Appendix.
\begin{lemma}
\label{lemma: valid}
Let $P' \subseteq \R^J \times \R^N$ be a polyhedron satisfying $\conv{C'(N)} \subseteq P' \subseteq Z(N)$.
For any $\emptyset \neq S \subseteq N$, a valid inequality for $\conv{P' \setminus\interior{U'(S)}}$ is itself a valid inequality for $\conv{C'(N)}$.
\end{lemma}
Finally, we prove that on each iteration, either no blocking coalition exists or progress can be made.
\begin{theorem}
\label{theorem: intersection cuts}
Let $P' \subseteq \R^J \times \R^N$ be a polyhedron satisfying $\conv{C'(N)} \subseteq P' \subseteq Z(N)$ and $(x^*, u^*) \in P'$ be an extreme point.
If $u^* \notin C(N)$, then there exists an intersection cut that strengthens $P'$ into a polyhedron $\hat{P}'$ satisfying $\conv{C'(N)} \subseteq \hat{P}' \subsetneq P' \subseteq Z(N)$.
\end{theorem}
\begin{proof}{Proof}
Since $P' \subseteq Z(N)$, it follows that $(x^*, u^*) \in Z(N)$.
Therefore, $u^* \in \projection{u}{Z(N)} = U(N)$.
Then, since $u^* \notin C(N)$, there must exist some $\emptyset \neq S \subseteq N$ such that $u^* \in \interior{U(S)}$.
By Lemma~\ref{lemma: U - U'}, this implies that $(x^*, u^*) \in \interior{U'(S)}$.
In turn, Lemma~\ref{lemma: conv(P'-int(U'(S))} implies that $(x^*, u^*) \notin \conv{P' \setminus \interior{U'(S)}}$, so that there exists an intersection cut that cuts $(x^*, u^*)$.
Lastly, Lemma~\ref{lemma: valid} implies the cut is valid for $\conv{C'(N)}$, so that it strengthens $P'$ into a polyhedron $\hat{P}'$ satisfying satisfying $\conv{C'(N)} \subseteq \hat{P}' \subsetneq P' \subseteq Z(N)$.
\end{proof}

To operationalize Theorem~\ref{theorem: intersection cuts}, first note that \eqref{eq: membership} and Lemma~\ref{lemma: membership} can be used to determine whether $u^* \notin C(N)$.
Second, if we find that $(x^*, u^*) \in \interior{U'(S)}$ for some $\emptyset \neq S \subseteq N$, then for each corresponding (non-basic variable) extreme ray $r \in \mathcal{R}$ we can compute \eqref{eq: zeta_r} as
\begin{equation}
\label{eq: zeta_r S}
    \zeta_r = \max \left\{ \zeta \geq 0: (x^*, u^*) + \zeta r \in U'(S) \right\},
\end{equation}
formulated as a linear program.
The full method for a linear objective is summarized in Algorithm~\ref{alg: algorithm}.
\begin{algorithm}[ht]
    \SetAlgoNoEnd
	\KwIn{NTU LP game $\left(A, \{b^i\}_{i \in N}, \{v^i\}_{i \in N}\right)$, linear objective $g: \R^J \times \R^N \rightarrow \R$, objection tolerance $\delta \geq 0$}
	\KwOut{$(x^*, u^*)$}
    $P' \gets Z(N)$ \tcp*{Specified by $A$, $\{b^i\}_{i \in N}$ and $\{v^i\}_{i \in N}$}
    $(x^*, u^*) \gets \argmax_{(x, u) \in P'} g(x,u)$\;
    \While{$\epsilon(u^*) > \delta$}{
        $S \gets \textrm{blockingCoalition}(u^*)$ \tcp*{Generated using \eqref{eq: membership}}
        $P' \gets P' \cap \textrm{intersectionCut}(P', U'(S), (x^*, u^*))$ \tcp*{Generated using \eqref{eq: zeta_r S}}
        $(x^*, u^*) \gets \argmax_{(x, u) \in P'} g(x,u)$\;
    }    
    \textrm{return} $(x^*, u^*)$\;
	\caption{Cutting plane algorithm for a linear objective over $\conv{C'(N)}$.}
	\label{alg: algorithm}
\end{algorithm}

The finite convergence of Algorithm~\ref{alg: algorithm} is not immediate.
For example, \citet{balas1971intersection} proved the finite convergence of an intersection cut algorithm for integer programs based on a cut integerization technique that is however not applicable in our setting.
\citet[Theorem 3.1]{porembski2001finitely} showed that intersection cut algorithms have finite convergence if the distance between a point being cut and the corresponding cutting plane (i.e., the ``depth'' of the cut) is bounded from below.
This is the case for any practical implementation of Algorithm~\ref{alg: algorithm} with objection tolerance $\delta > 0$.
We may assume that $P'$ is bounded after adding the cuts corresponding to singleton coalitions.
Together with the fact that the \texttt{while} loop passes only if $\epsilon(u^*) > \delta$, we find that the depth of a cut cannot be arbitrarily small.
That is, for any fixed $\delta > 0$, after a finite number of iterations, Algorithm~\ref{alg: algorithm} returns a point $(x^*, u^*) \in Z(N)$ that is approximately in $C'(N)$ in the sense that no coalition can improve upon $u^*$ beyond an additive $\delta$ for each of its members.

As we illustrate in the following continuation of the original instance in Example~\ref{ex: insufficient}, in the event the core (or its relaxation by the choice of $\delta$) is empty, Algorithm~\ref{alg: algorithm} certifies it by producing an infeasible system of linear inequalities.
\begin{example}[Example~\ref{ex: insufficient}, continued.]
\label{ex: mis}
The Pareto-optimal extreme points in $Z(\{1,2,3\})$ are $z = (\underbrace{3,0}_{x_1, \ x_2},\underbrace{2,2,-2}_{u_1, \ u_2, \ u_3})$ and $z' = (\underbrace{0,3}_{x_1, \ x_2}, \underbrace{1,1,1}_{u_1, \ u_2, \ u_3})$.
Therefore, $C(\{1,2,3\}) \subseteq \conv{\{z, z'\}} \coloneqq \{(1 - \eta) z + \eta z' : 0 \leq \eta \leq 1\}$.
The singleton coalition $S = \{3\}$ requires $u_3 \geq \frac{1}{3}$, which induces the constraint $\eta \geq \frac{7}{9}$.
The resulting apex at $\eta = \frac{7}{9}$, namely $z^* = (\frac{2}{3}, \frac{7}{3}, \frac{11}{9}, \frac{11}{9}, \frac{1}{3})$, satisfies $z^* \in \interior{U'(\{1,2\}}$ since it is strongly dominated by the point $u^* = (\cdot,\cdot,\frac{4}{3}, \frac{4}{3}, \cdot) \in U'(\{1,2\})$ achieved if the players in $S = \{1, 2\}$ invest their two units of budget on $x_1$.
In turn, the extreme ray along the line from $z^*$ to $z$ induces the incompatible constraint $\eta \leq \frac{2}{3}$.
Equivalently, if represented as a system of linear inequalities, these intersection cuts on $Z(\{1,2,3\})$ result in the following irreducible infeasible sub-system: \textit{i)} $x_1 + x_2 \leq 3$, \textit{ii)} $u_3 \leq -\frac{2}{3}x_1 + \frac{1}{3} x_2$, \textit{iii)} $u_3 \geq \frac{1}{3}$, and \textit{iv)} $x_1 \geq 1$.
\end{example}

We conclude this section by expanding on the reason for working in the design-utility space $\R^J \times \R^N$ rather than the utility space $\R^N$, where the core is originally defined.
Recall the sets $\{U(S)\}_{\emptyset \neq S \subseteq N}$ are represented implicitly, namely as the projections $U(S) = \projection{u}{Z(S)}$ rather than by explicit systems of linear inequalities (and, by Example~\ref{ex: cyclic}, such representations could be sizable).
First, the lack of explicit representations complicates the practicality of generating intersection cuts in the utility space, as access to the simplex tableau is necessary.
Second, the general strategy of working in an extended space has the further advantage that it can accommodate the optimization of objectives that are non-linear in the design-utility space, as long as they can be linearized through an extended polyhedral formulation.
For example, to solve $\max_{(x, u) \in C'(N)} \min_{i \in N} u_i$ one can introduce a new variable $w$ constrained by $w \leq u_i$ for all $i \in N$.
By generating cuts as in Algorithm~\ref{alg: algorithm} but for this higher dimensional relaxation, one ultimately obtains a higher dimensional vertex whose projection into the utility space is the desired optimizer.

\section{Case Study: Frequency Setting in the Chicago Bus System}
\label{sec: frequency setting in the chicago bus system}

In this section we revisit the ridership versus coverage dilemma introduced in Section~\ref{sec: motivating example} through a data-driven case study.
In Section~\ref{sec: model} we describe our model of transit line frequency setting and its corresponding NTU LP game.
In Section~\ref{sec: data} we describe the datasets we use, while in Section~\ref{sec: practical implementation} we describe our practical implementation of Algorithm~\ref{alg: algorithm}.
Finally, in Section~\ref{sec: operational and distributive implications} we examine the operational and distributive implications of cooperation on transit service plans.

\subsection{Model}
\label{sec: model}

Let $G = (V, E)$ be an undirected graph representing a city's road network, where the nodes $V$ correspond to intersections and the edges $E$ correspond to street segments.
The edges are weighted by their length $\ell: E \rightarrow \mathbb{Q}_{\geq 0}$ in meters.
Suppose a transit agency operates a set $J$ of bus lines, such as the \textcolor{orange}{A} and \textcolor{blue}{B} lines from Figure~\ref{fig: dilemma}.
More formally, a line is defined as a simple walk in $G$.
For each line $j \in J$, let $s^j \subseteq V$ denote its fixed set of stops where passengers can board or alight.
Moreover, let $a \in \mathbb{Q}_{\geq 0}^{J}$ be a vector encoding the lengths of the lines in kilometers ($\mathrm{km}$).

The agency serves a set $R$ of riders from a set $N$ of politically organized geographic regions, such as city council districts within a municipality or counties within a metropolitan area; the regions (but not the individual riders) form the set of players in the game.
Each rider $r \in R$ has a fixed origin $o^{r} \in V$ and destination $d^{r} \in V$, and they value line $j$ to the extent to which it is accessible at both ends of their trip; in this model, no line transfers are allowed.
To this effect, let $\ell(o^{r},s^j) \in \mathbb{Q}_{\geq 0}$ denote the distance between $o^{r}$ and the closest stop in $s^j$ in meters.
We model the accessibility of line $j$ with respect to $o^{r}$ through the piece-wise linear function
\begin{equation*}
    v_j^{o^{r}} = 
    \begin{cases}
    1, & \text{if $\ell(o^{r},s^{r}) < 400$} \\
    1 - \frac{\ell(o^{r},s^j) - 400}{1600 - 400}, &  \text{if $400 \leq \ell(o^{r},s^j) \leq 1600$} \\
    0, & \text{if $\ell(o^{r},s^{r}) > 1600$,}
    \end{cases}
\end{equation*}
where the threshold distances of $400$ and $1600$ meters are taken from the transportation literature on the topic of accessibility and willingness to walk~\citep{walker2024human}.
We model accessibility with respect to $d^{r}$ similarly.
Then, the overall value rider $r$ has for line $j$ is its accessibility with respect to $o^{r}$ or $d^{r}$, whichever is worst.
That is,
\begin{equation*}
    v_j^{r} = \min\{v_j^{o^{r}}, v_j^{d^{r}}\}.
\end{equation*}
In this way, we form a vector $v^r \in \mathbb{Q}_{\geq 0}^J$ that encodes the value the rider has for each line.
In turn, each region $i \in N$ is concerned with riders whose origin or destination are within its boundaries.
It aggregates these riders' roundtrip valuation vectors into a single, average valuation vector
\begin{equation}
\label{eq: valuation model}
    v^i 
    = \frac{\sum_{r \in R} \left(\mathbf{1}_{o^r \in \Gamma(i)} + \mathbf{1}_{d^r \in \Gamma(i)} \right)v^r}{\sum_{r \in R} \left(\mathbf{1}_{o^r \in \Gamma(i)} + \mathbf{1}_{d^r \in \Gamma(i)} \right)},
\end{equation}
where $\mathbf{1}_{\{\cdot\}}$ is the indicator function and $\Gamma(i)$ is its geographic area.
Therefore, $N$ can be interpreted as a set of representatives (e.g., city council members) who act on behalf of the collective interests and resources of their constituents, in this case transit riders.

We assume the agency funds its operations through fare-box recovery.
Each rider $r$ contributes a flat rate of $\$1$ to the system on each direction of travel, so that 
\begin{equation}
\label{eq: budget model}
    b^i = \sum_{r \in R} \left(\mathbf{1}_{o^r \in \Gamma(i)} + \mathbf{1}_{d^r \in \Gamma(i)}\right)
\end{equation}
is the total fare amount paid within region $i$ and $b(N) = \sum_{i \in N} b^i = 2|R|$ is the total fare amount paid in the entire service area.
Since ridership correlates with population, this budget model can also be interpreted as a form of population-based budgetary apportionment.

We assume the frequency at which a line operates can be measured in terms of $\mathrm{\$/km}$.
That is, the frequency of a line is proportional to the budget apportioned to it, and inversely proportional to its length.
In this way, the agency's frequency design space \eqref{eq: X(N)} is given by $X(N) = \{x \in \mathbb{R}_{\geq 0}^J : a^T x \leq b(N)\}$.
If a subset of regions $\emptyset \neq S \subseteq N$ cooperate to form their own transit agency (e.g., embodying the real-world examples in Section~\ref{sec: motivating example}), their frequency design space \eqref{eq: X(S)} is given by $X(S) \coloneqq \{x \in \R_{\geq 0}^J : a^Tx \leq b(S) \}$.
Following \eqref{eq: Z(S)} and \eqref{eq: U(S)}, the region-line valuation model \eqref{eq: valuation model} in turn gives rise to their design-utility space $Z(S)$ and their utility space $U(S) = \projection{u}{Z(S)}$.
This completes the specification of the NTU LP game $\{U(S)\}_{\emptyset \neq S \subseteq N}$;
by Theorem~\ref{theorem: balanced}, $C(N) \neq \emptyset$.

In our experiments we distinguish whether the agency pursues a strictly ``maximum coverage'' or ``maximum ridership'' objective (refer to Section~\ref{sec: motivating example} and Figure~\ref{fig: dilemma}).
We approximate the ridership objective through an utilitarian conceptualization of social welfare.
That is, we assume that in this case the agency would operate a service plan $x^* \in X(N)$ such that
\begin{equation}
\label{eq: utilitarian}
    (x^*, u^*) \in \argmax_{(x, u) \in Z(N)} \frac{1}{b(N)} \sum_{i \in N} b^i u_i.
\end{equation}
Similarly, we approximate the coverage objective through a maximin conceptualization of social welfare.
We assume that in this case the agency would operate a service plan $x^* \in X(N)$ such that
\begin{equation}
\label{eq: maximin}
    (x^*, u^*) \in \argmax_{(x, u) \in Z(N)} \min_{i \in N} u_i.
\end{equation}
To implement \eqref{eq: maximin}, we employ an extended formulation that linearizes the maximin objective, as described in Section~\ref{sec: cutting plane algorithm}.
We break ties in \eqref{eq: maximin} in favor of solutions that maximize \eqref{eq: utilitarian} as a secondary objective by adding it to the objective function with a small multiplicative factor.

\subsection{Data}
\label{sec: data}

Our implementation of the model described in Section~\ref{sec: model} is based on publicly available datasets from Cook County, IL, whose county seat is the city of Chicago.
We obtain a road network on $96,413$ nodes and $147,430$ edges using \texttt{osmnx}~\citep{boeing2017osmnx}.
We map the stop sequence of all bus lines in the county, queried from OpenStreetMap~\citep{osm}, to the nearest nodes in the network using their geographical coordinates.
We preprocess the set of lines to obtain their route lengths, and to remove any duplicate or exceedingly short lines (those involving $5$ stops or less).
We are ultimately left with $499$ different lines.

We let $N$ be Chicago's 77 community areas; this is an official geographic division of the city that is commonly used for statistical purposes.
Its shapefiles are available through the \cite{cdp}.
In Figure~\ref{fig: lines}, we plot the different transit lines overlaid with the community areas.
To approximate the rider demand for the bus system, we consider all ride-hailing trips dated July 1, 2024 available through the \citet{cdp}.
We discard any trip whose line valuation model \eqref{eq: valuation model} is identically zero for all lines (i.e., trips that are too far from all lines), whose duration was less than a minute, whose length was less than $800$ meters, or whose endpoints were at the Chicago O'Hare International Airport or the Chicago Midway International Airport (which are primarily served by heavy rail rather than buses).
We label the 70,524 remaining trips with their community areas of origin and destination to instantiate the rider set $R$ and the models \eqref{eq: valuation model} and \eqref{eq: budget model} for $N$.
We represent the distribution of community area budgets in space in Figure~\ref{fig: budgets}.

\begin{figure}[ht]
    \centering
    \caption{
        Illustration of the model data overlaid in space with Chicago's 77 community areas.
        In Figure~\ref{fig: lines}, the $499$ different bus lines are shown in various colors.
        In Figure~\ref{fig: budgets}, each community area's budget is represented by a black circle centered at the area's ridership-weighted centroid and with radius logarithmic in $b^i$.
    }
    \label{fig: data}
    \begin{subfigure}{0.3\linewidth}
        \centering
        \includegraphics[width=0.75\linewidth]{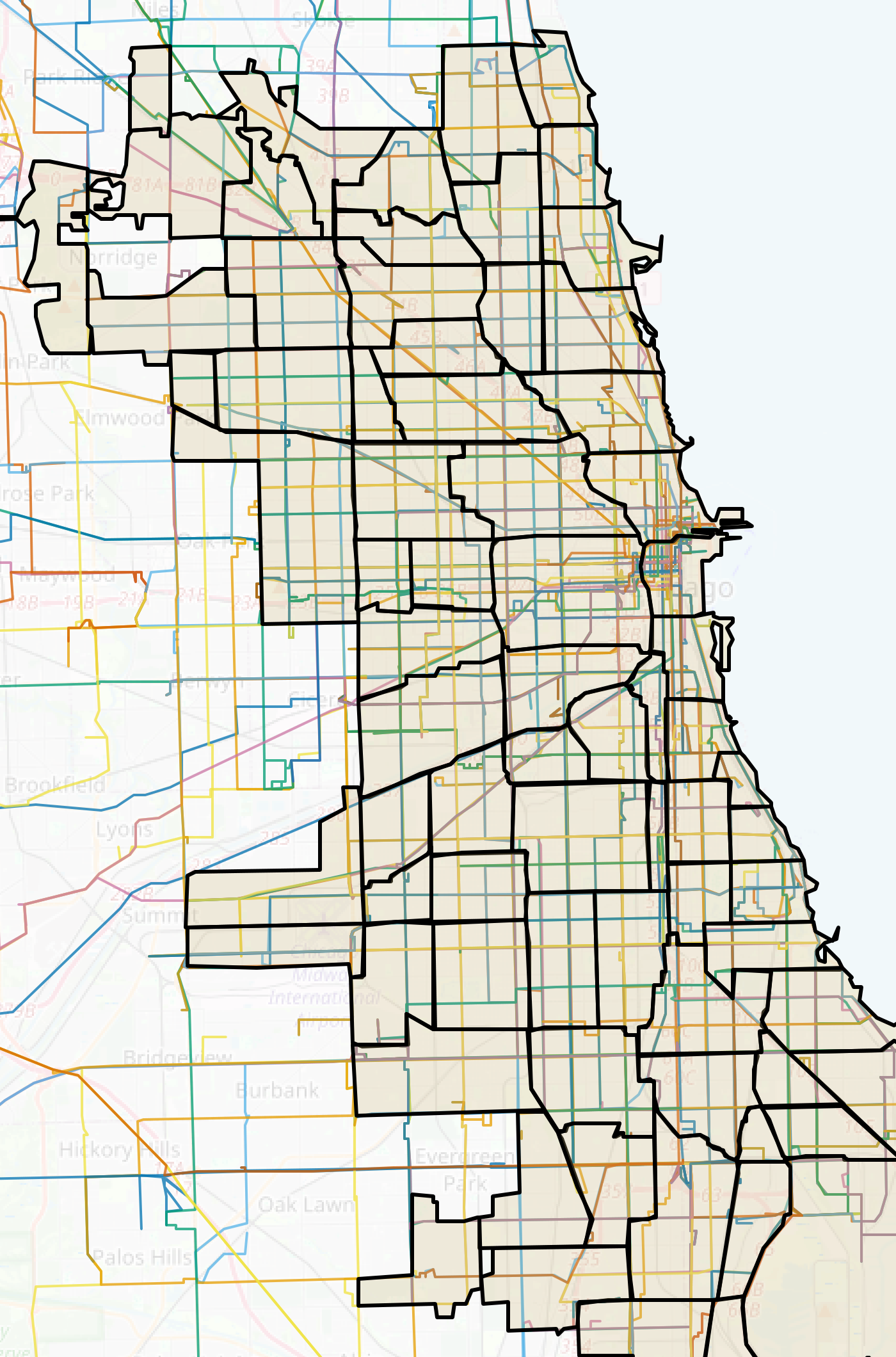}
        \caption{Bus lines $J$.}
        \label{fig: lines}
    \end{subfigure}
    \begin{subfigure}{0.3\linewidth}
        \centering
        \includegraphics[width=0.75\linewidth]{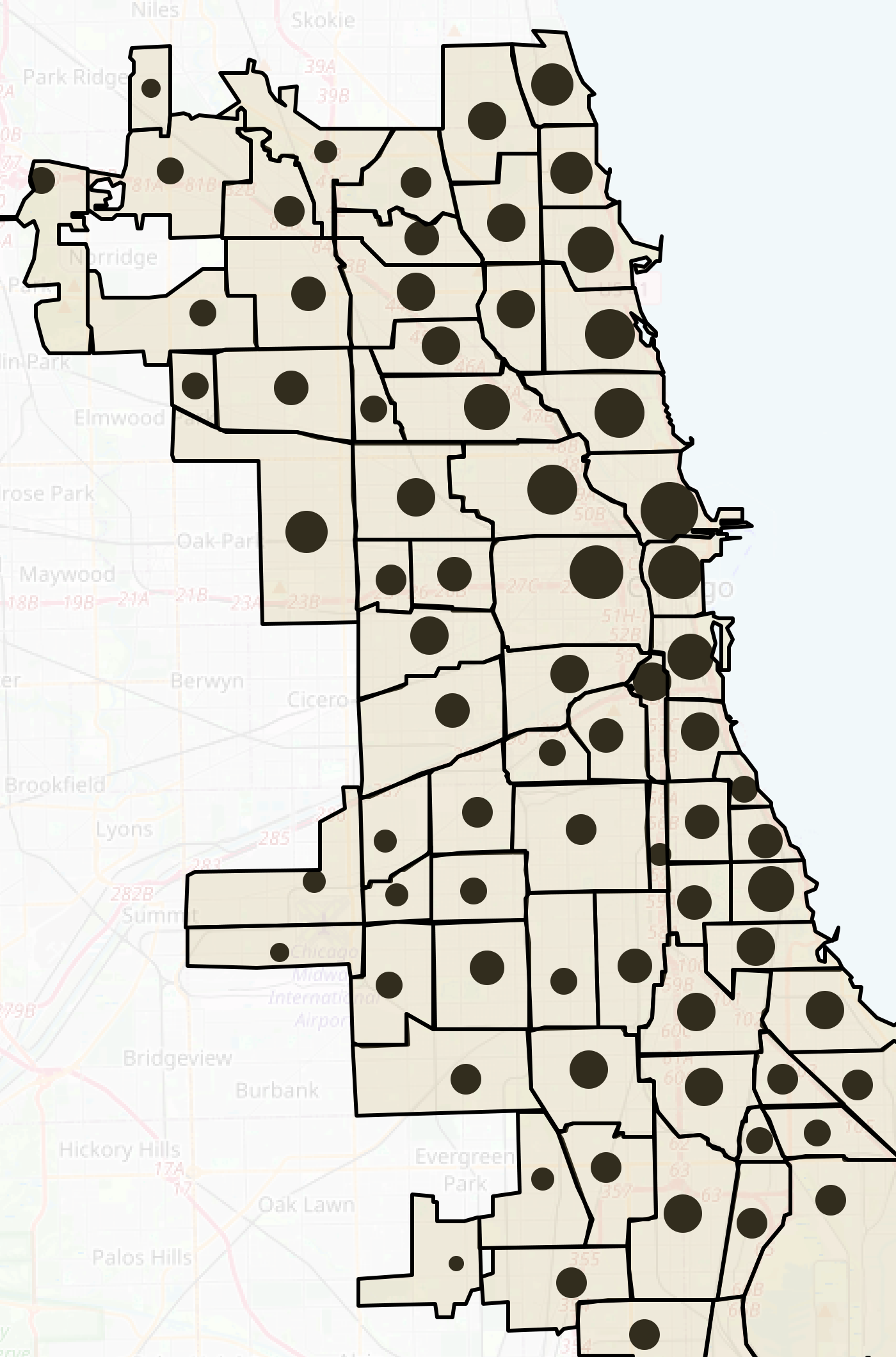}
        \caption{Budgets $b^i$ (on a $\log_2$ scale).}
        \label{fig: budgets}
    \end{subfigure}
\end{figure}
Note the dimensions of the resulting NTU LP game: $|K| = 1$, $|J| = 499$, and $|N| = 77$.

\subsection{Implementation}
\label{sec: practical implementation}

We run Algorithm~\ref{alg: algorithm} for $100$ iterations in conjunction with various practical improvements.
Pure cutting plane algorithms such as Algorithm~\ref{alg: algorithm} are known to exhibit challenges with numerical stability and speed of convergence;
managing these challenges is non-trivial: in integer programming for example, they can be alleviated by implementing cutting planes in combination with branch-and-cut, or by using a lexicographic version of the dual simplex method~\citep{zanette2011lexicography}.
However, these kinds of improvements are not readily available to us because our description of the core does not involve integer variables.
Therefore, we design problem-specific practical improvements that aim to alleviate two qualitatively different kinds of challenges: good cuts are computationally expensive to generate, and they can become increasingly parallel.

As a first improvement, we add the individual rationality (IR) constraints $u_i \geq b^i \max_{j \in J} \{v_j^i / a_j\} > 0$ for all $i \in N$ since these are simple to produce and are part of the definition \eqref{eq: C(N)} of the core.
Second, as a major enhancement, we adapt \eqref{eq: membership} to optimize for a multiplicative least objection rather than an additive least objection.
Specifically, we replace \eqref{cnstr: epsilon} by 
\begin{equation}
\label{eq: epsilon revised}
    y_i^S \left(\epsilon - u_i^S/u_i^*\right) \leq 0, \quad \forall i \in N,
\end{equation}
which are well-defined since the individual rationality constraints ensure $u_i^* > 0$ for all $i \in N$.
In this revised version of \eqref{eq: membership}, $\epsilon(u^*) \geq 1$ since setting $y_i = 1$ and $u_i^S = u_i^*$ for all $i \in N$ induces a feasible solution with $\epsilon = 1$.
Therefore, testing whether $\epsilon(u^*) > 1$ is a valid core membership test akin to Lemma~\ref{lemma: membership}.
Our computational experience is that optimizing for multiplicative least objections produces a richer set of blocking coalitions throughout the execution of the algorithm.
Moreover, multiplicative least objections are more interpretable since $\epsilon(u^*)$ becomes a unit-less quantity: it represents the least utility improvement available to any member of a blocking coalition as a factor of their incumbent utility.
To improve numerical stability compared to using a big-$M$, we implement \eqref{eq: epsilon revised} using the solver's dedicated indicator constraint feature~\citep{gurobi2025}.

As another major enhancement, we maintain the collection $\mathcal{S} = \{S^1, S^2, \ldots\} \subseteq 2^N$ of blocking coalitions generated throughout the execution of the algorithm.
In each new iteration, we use this collection in two different ways:
\begin{enumerate}
    \item 
    We instantiate \eqref{eq: membership} with the incumbent $\mathcal{S}$ as a warm start collection of candidate coalitions.
    For each $S \in \mathcal{S}$, we give the solver the hint $y_i = 1$ for $i \in S$ and $y_i = 0$ for $i \in N \setminus S$.  
    In our implementation, we give the solver a timeout of $5$ minutes, at which point the incumbent solution is adopted and its blocking coalition is added to $\mathcal{S}$.
    Although this blocking coalition need not be optimal if the timeout comes into effect, it suffices to produce an intersection cut.
    \item 
    Cutting plane algorithms tend to perform better if cuts are introduced in rounds rather than sequentially.
    Therefore, in addition to cuts derived from the newly identified blocking coalitions, we derive cuts from coalitions in $\mathcal{S}$ that are also blocking coalitions in the current iteration.
    To alleviate numerical instabilities, we discard any such cut whose smallest-to-largest coefficient ratio is under $10^{-9}$, or whose corresponding coalition intersects the newly identified blocking coalition or has multiplicative least objection under $1 + 10^{-2}$.
\end{enumerate}

As a final improvement, in each iteration we temporarily augment the warm start collection $\mathcal{S}$ with candidate coalitions identified through a combinatorial heuristic that depends on the incumbent $u^* \in U(N)$.
The heuristic considers the lines $j \in J$ one at a time and sorts the players $i \in N$ in decreasing order of valuation $v_j^i$.
In this order, it finds the prefix coalition $\emptyset \neq S \subseteq N$ that achieves the largest multiplicative least objection (with respect to $u^*$) assuming its entire budget is dedicated to the operation of line $j$.
We add said $S$ to the warm start collection of the current iteration if its multiplicative least objection is greater than $1$.

In short, these collective adjustments aim to efficiently produce distinct, deep, and numerically stable cuts.
Following the solver's guidelines for handling numerical issues, we use the dual simplex method as our linear programming algorithm and set solver parameters that shift its focus toward more careful numerical computations~\citep{gurobi2025}.

\subsection{Operational and Distributive Implications of Cooperation}
\label{sec: operational and distributive implications}

Figures~\ref{fig: maximin service plans} and \ref{fig: utilitarian service plans} illustrate the progression of Algorithm~\ref{alg: algorithm} under the maximin \eqref{eq: maximin} and utilitarian \eqref{eq: utilitarian} service goals, respectively.
In these figures, the width of a line is logarithmic in its frequency.

As shown in Figure~\ref{fig: maximin inital service plan}, the initial maximin service plan (without any cooperation constraints) operates an expansive grid of suburban lines.
After IR constraints are introduced in Figure~\ref{fig: maximin ir service plan}, some suburban lines are dropped and frequency is increased downtown.
Then, the algorithm progresses as shown in Figure~\ref{fig: maximin iteration service plan}: as a blocking coalition is identified, a cut is added to break it apart by improving service for (at least some) of its members.
In the particular iteration shown, six contiguous community areas in the North Shore form a blocking coalition (left) that advocates for improved service along this corridor (right).
\begin{figure}[ht]
    \centering
    \caption{
        Progression of maximin service plans toward a cooperative solution.
        The different bus lines are color-coded, and the width of each line $j \in J$ is proportional to $\log_2(1 + x_j^*) \geq 0$ in the incumbent service plan.
    }
    \label{fig: maximin service plans}
    \begin{subfigure}{0.24\linewidth}
        \centering
        \includegraphics[width=\linewidth]{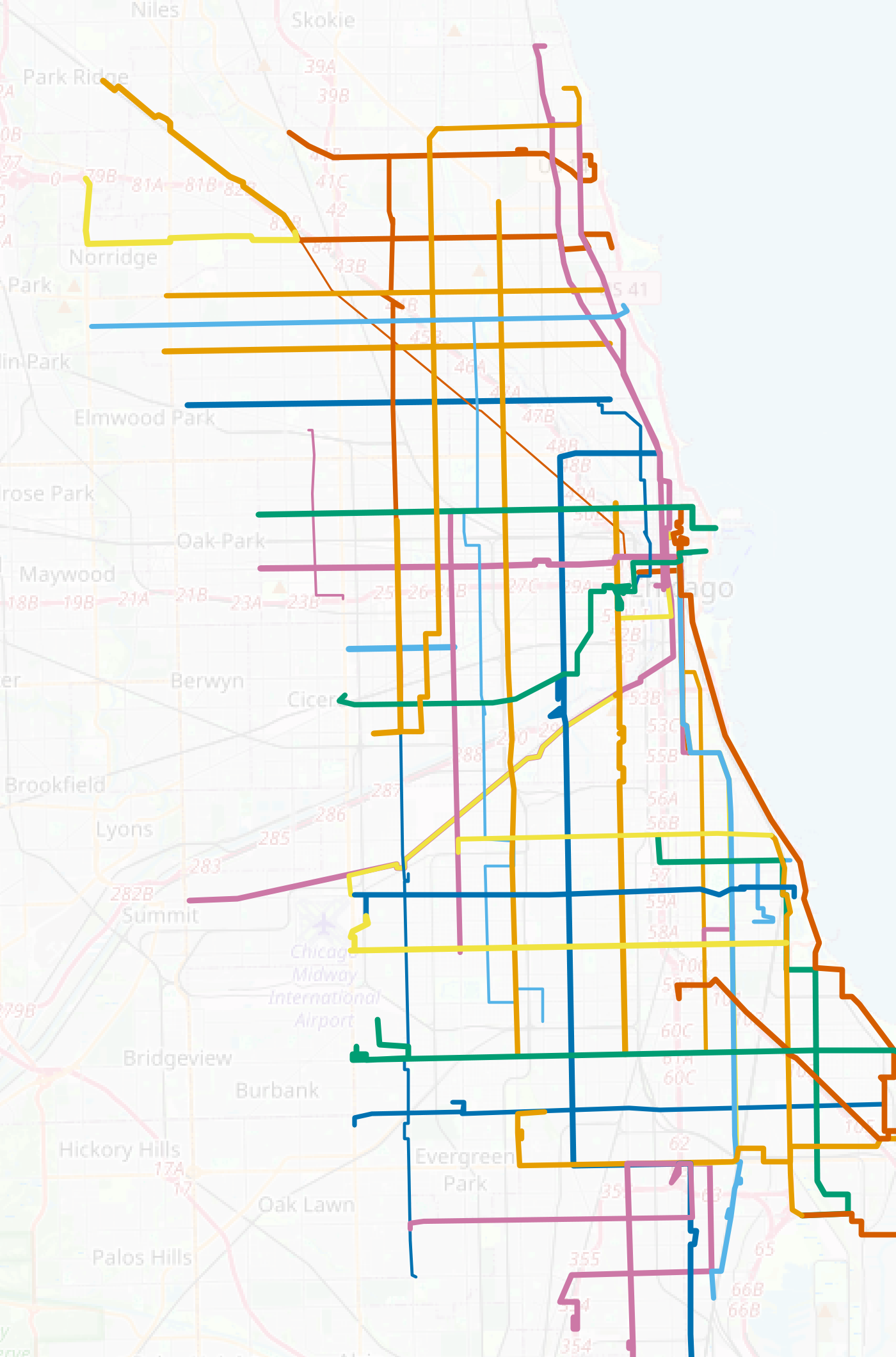}
        \caption{
            Initial service plan.
        }
        \label{fig: maximin inital service plan}
    \end{subfigure}
    \hfill
    \begin{subfigure}{0.24\linewidth}
        \centering
        \includegraphics[width=\linewidth]{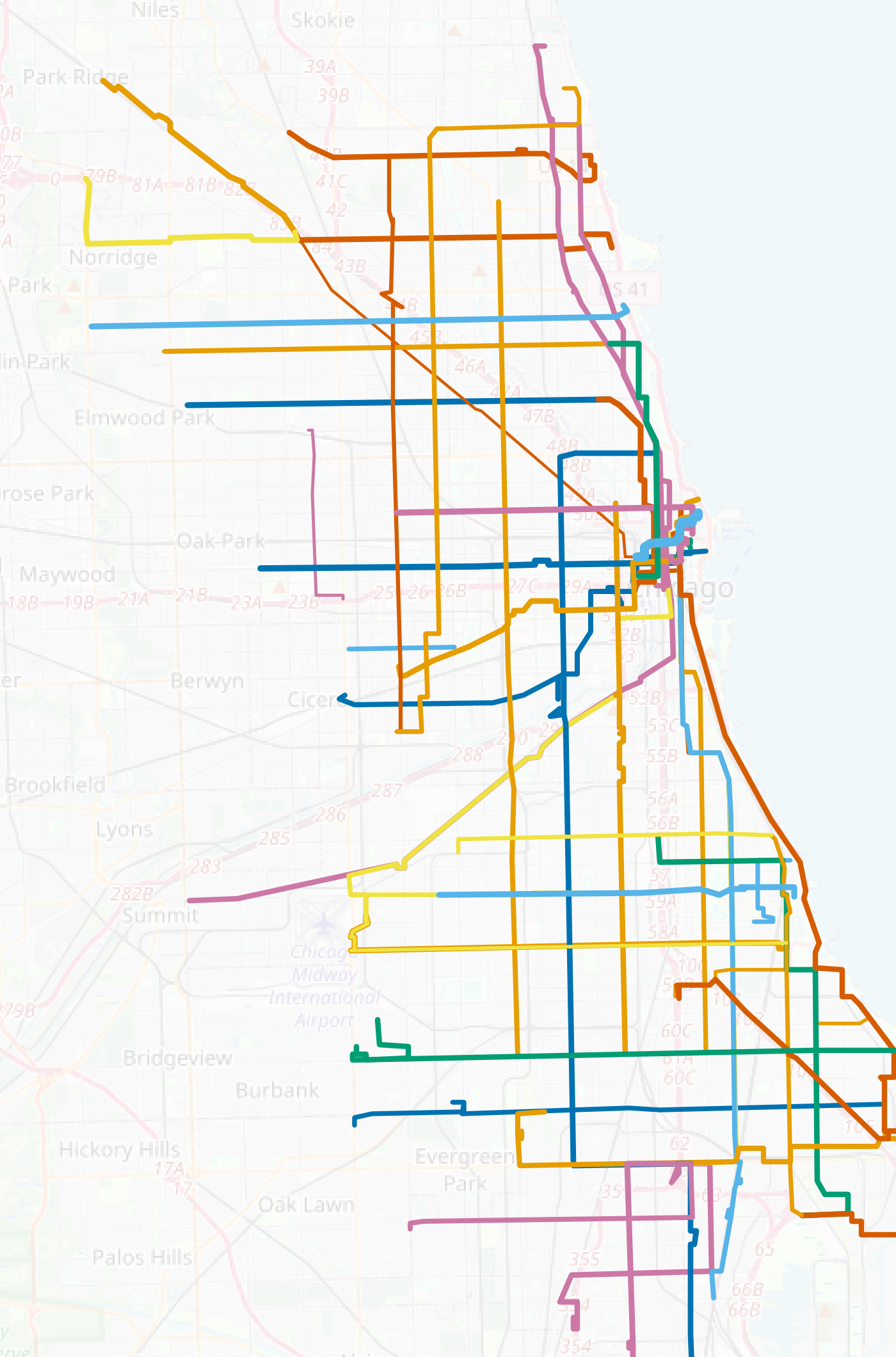}
        \caption{
            After IR constraints.
        }
        \label{fig: maximin ir service plan}
    \end{subfigure}
    \hfill
    \begin{subfigure}{0.5\linewidth}
        \centering
        \includegraphics[width=0.48\linewidth]{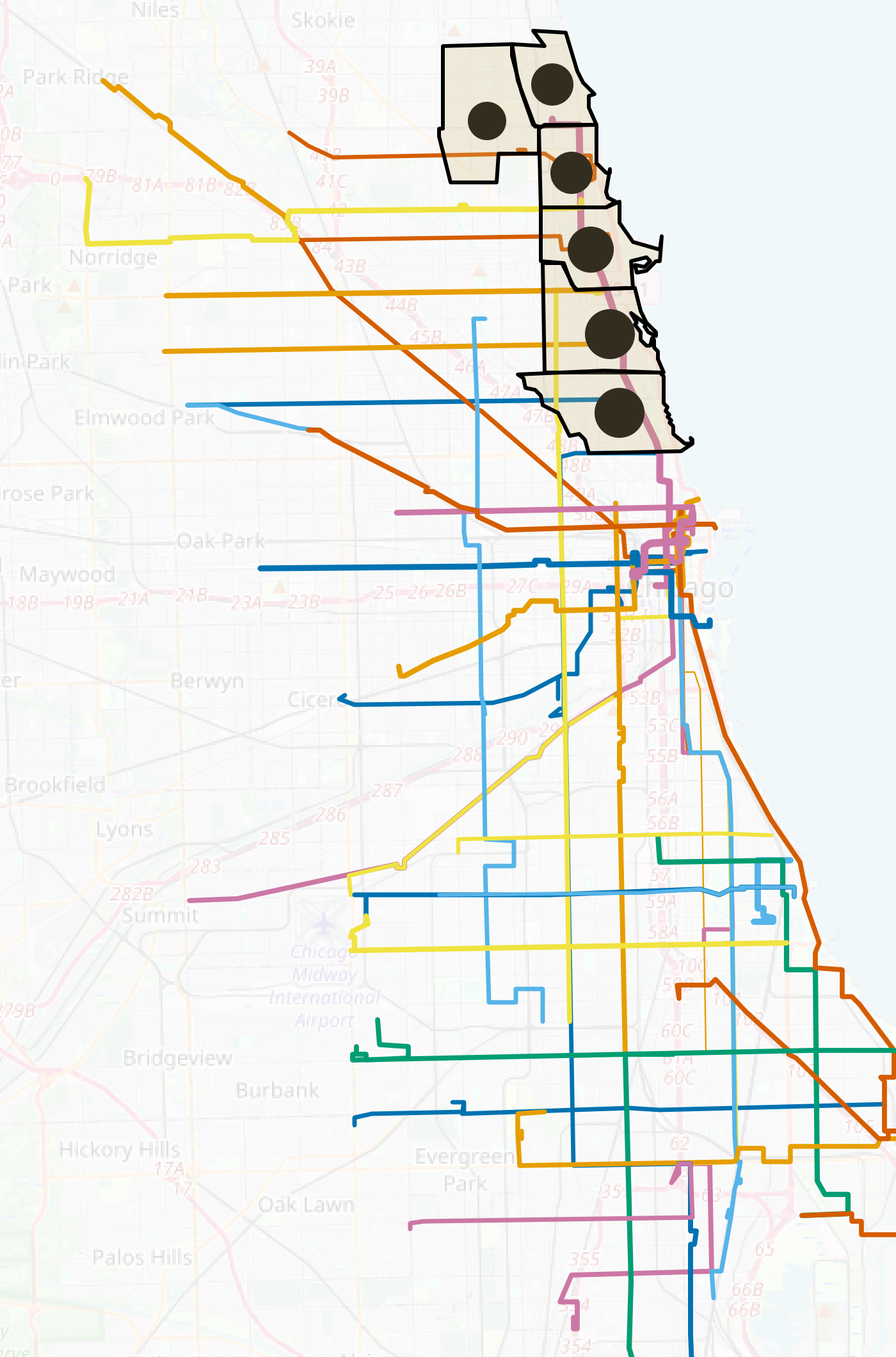}
        \includegraphics[width=0.48\linewidth]{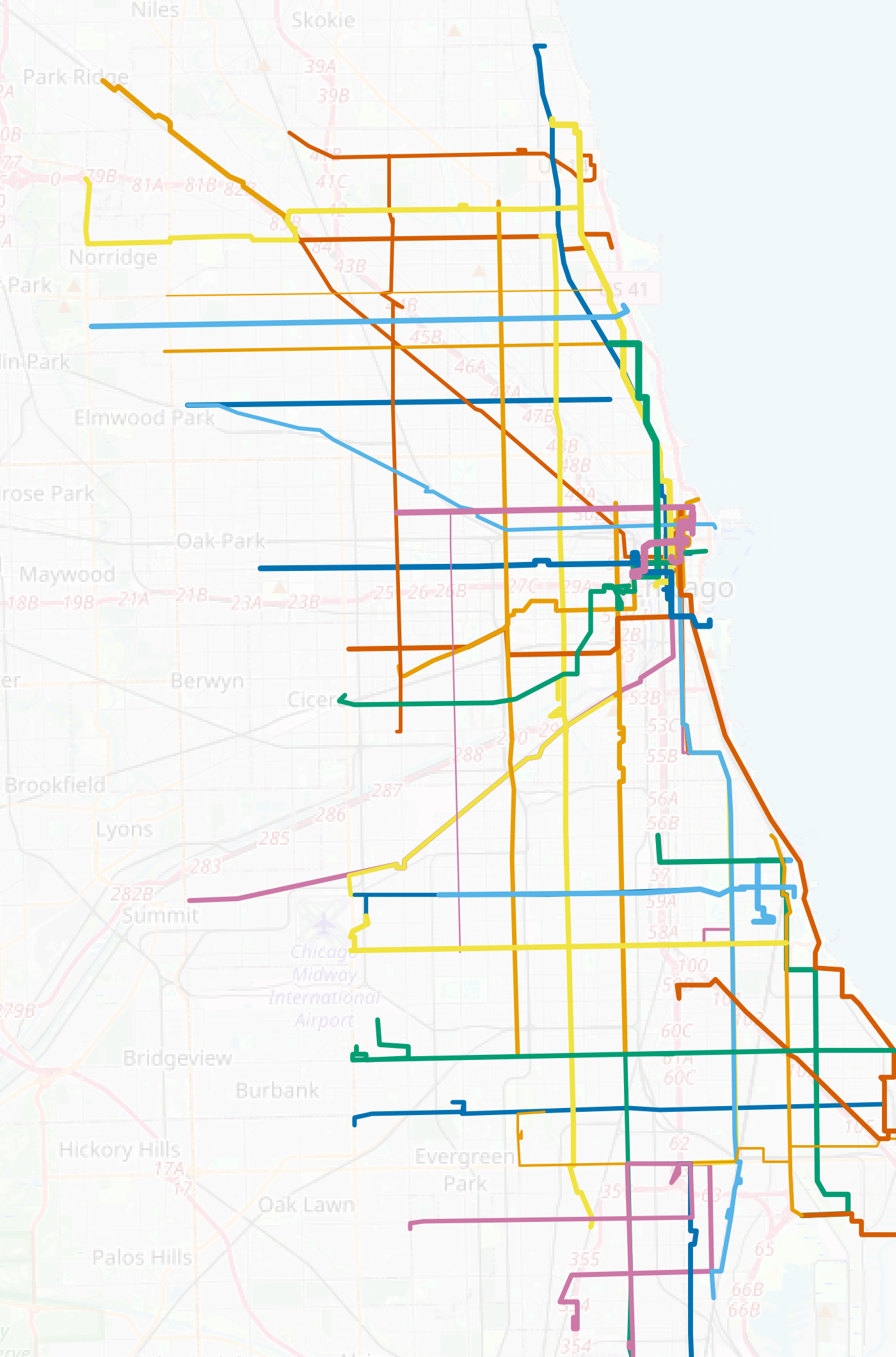}
        \caption{
            An iteration of the cutting plane algorithm.
        }
        \label{fig: maximin iteration service plan}
    \end{subfigure}
\end{figure}

The service plan progression under the utilitarian service goal is vastly different.
As seen in Figure~\ref{fig: utilitarian inital service plan}, the initial utilitarian service plan operates a single downtown line.
This unrealistic phenomenon is a polyhedral artifact: without any cooperation constraints, \eqref{eq: utilitarian} maximizes a linear objective function subject to a single linear constraint, so that $x^*$ is supported on a single line.
Once the IR constraints are introduced in Figure~\ref{fig: utilitarian ir service plan}, this artifact is eliminated and the service plan operates a sparse grid of suburban trunk lines, along with frequent service downtown.
The algorithm then progresses as shown in Figure~\ref{fig: utilitarian iteration service plan}.
In the particular iteration shown, twelve discontiguous community areas in the South Side and Far South Side form a blocking coalition (left) that advocates for a denser service network traversing a larger fraction of their geographical regions (right).

\begin{figure}[ht]
    \centering
    \caption{
        Progression of utilitarian service plans toward a cooperative solution.
        The different bus lines are color-coded, and the width of each line $j \in J$ is proportional to $\log_2(1 + x_j^*) \geq 0$ in the incumbent service plan.
    }
    \label{fig: utilitarian service plans}
    \begin{subfigure}{0.24\linewidth}
        \centering
        \includegraphics[width=\linewidth]{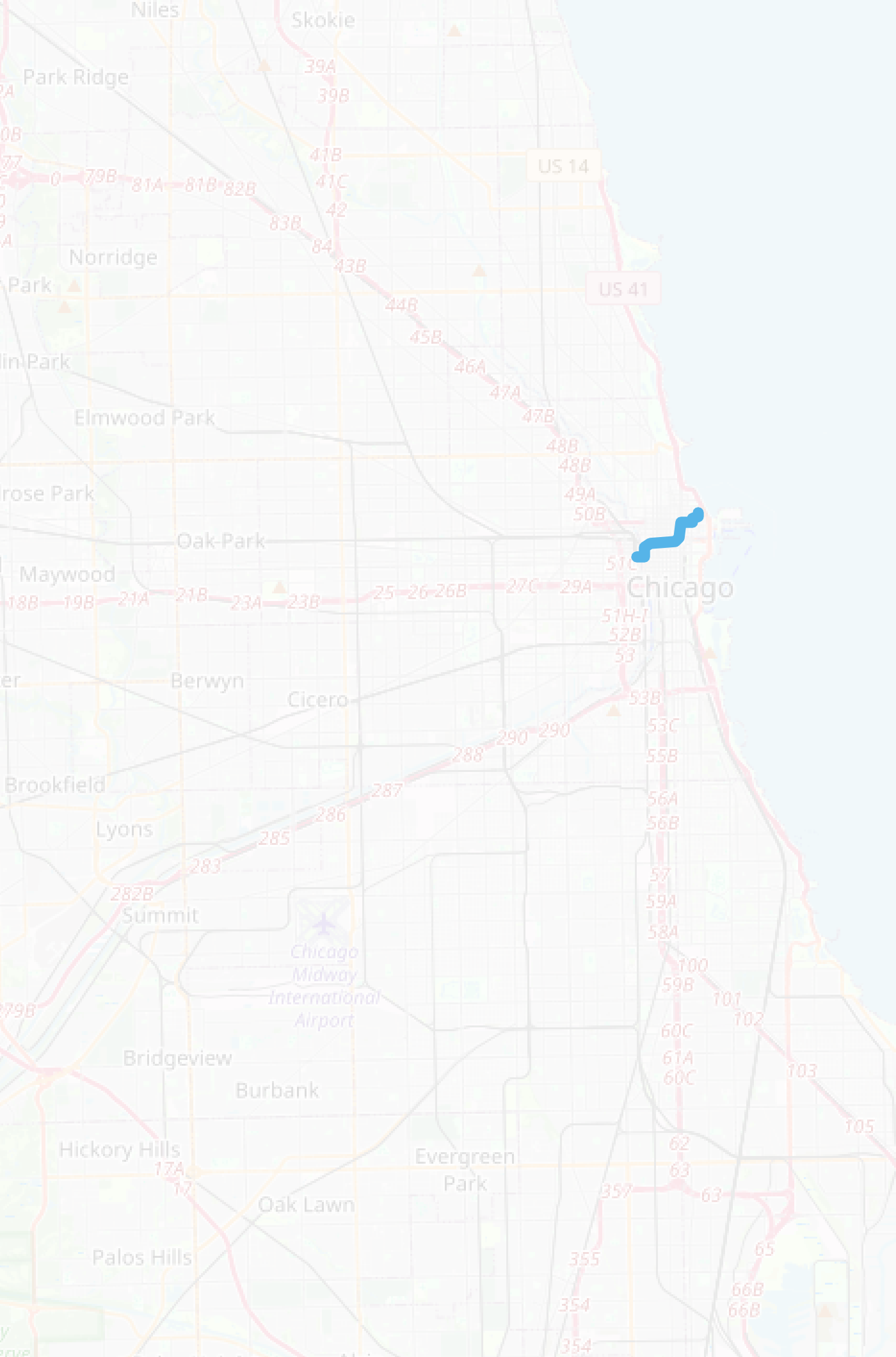}
        \caption{
            Initial service plan.
        }
        \label{fig: utilitarian inital service plan}
    \end{subfigure}
    \hfill
    \begin{subfigure}{0.24\linewidth}
        \centering
        \includegraphics[width=\linewidth]{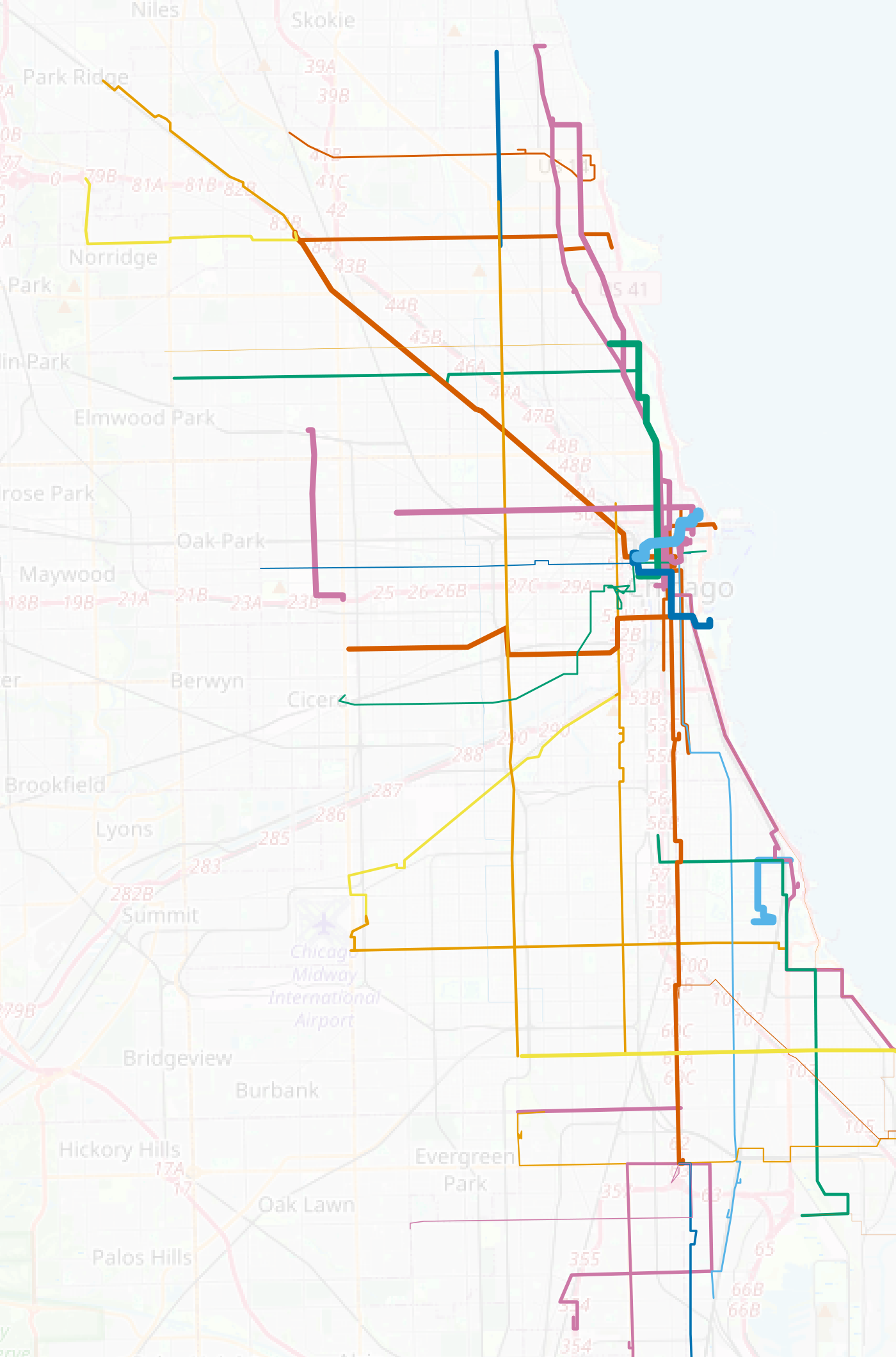}
        \caption{
            After IR constraints.
        }
        \label{fig: utilitarian ir service plan}
    \end{subfigure}
    \hfill
    \begin{subfigure}{0.5\linewidth}
        \centering
        \includegraphics[width=0.48\linewidth]{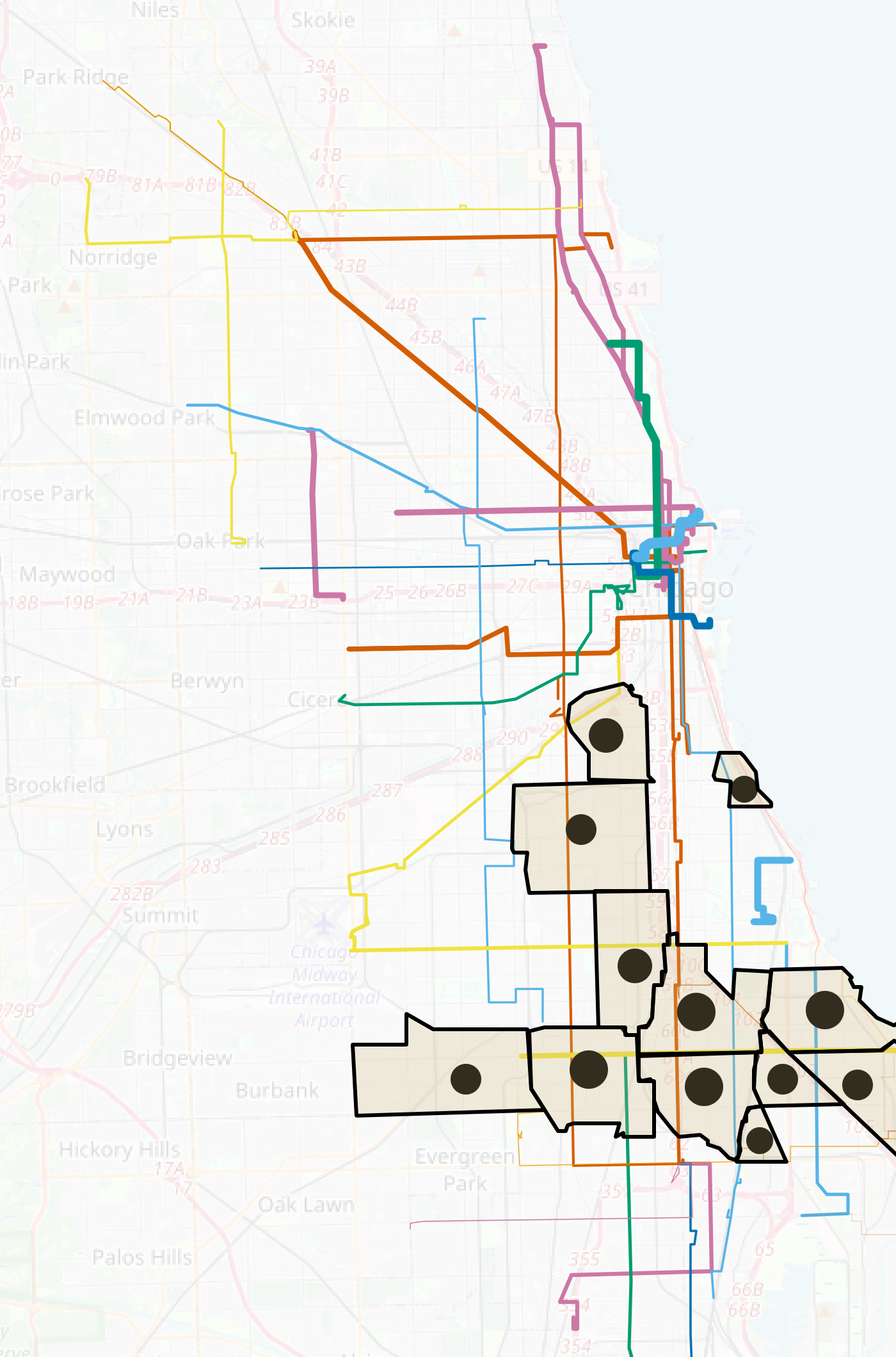}
        \includegraphics[width=0.48\linewidth]{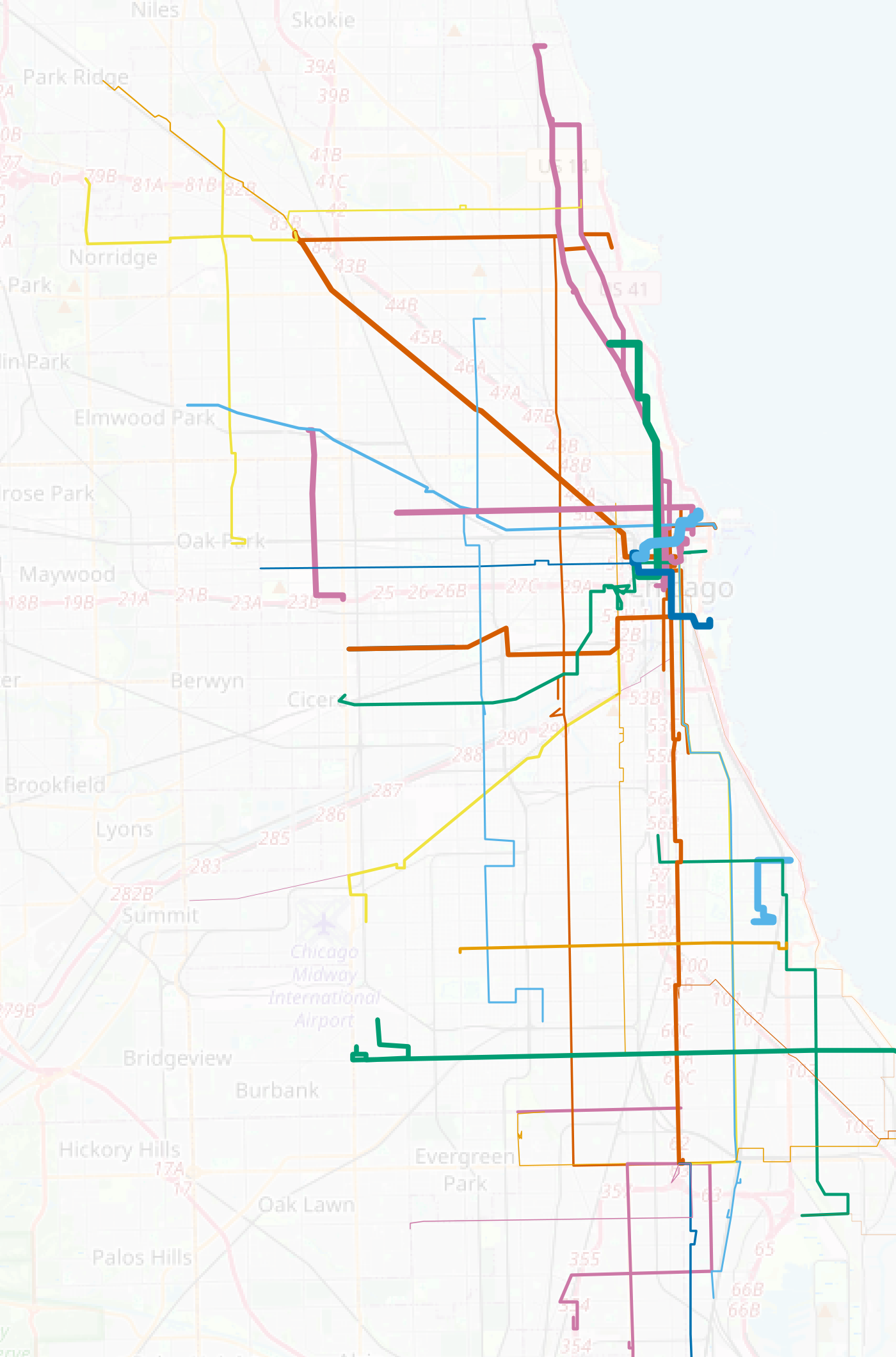}
        \caption{
            An iteration of the cutting plane algorithm.
        }
        \label{fig: utilitarian iteration service plan}
    \end{subfigure}
\end{figure}

We examine the distributive implications of regional cooperation in Figure~\ref{fig: distributive implications}, where we plot the individual rider utilities $\left\{\left(v^r\right)^Tx^*\right\}_{r \in R}$ given $x^* \in X(N)$ optimizing the maximin or utilitarian service goal, with or without cooperation constraints among the community areas.
Figure~\ref{fig: maximin distributive implications} shows that, under the maximin service goal, cooperation leads to a large utility gain for riders in the top utility quartile.
Conversely, Figure~\ref{fig: utilitarian distributive implications} shows that under the utilitarian service goal, cooperation leads to a large utility loss for riders in the top utility quartile. 
\begin{figure}[ht]
    \centering
    \caption{
        Distributive implications of cooperation under the maximin and utilitarian service goals.
        The utility distribution curves for the initial service plans (without cooperation) are shown in yellow, whereas the utility distribution curves for the final service plans (with cooperation) after $100$ iterations of Algorithm~\ref{alg: algorithm} are shown in blue.
    }
    \label{fig: distributive implications}
    \begin{subfigure}{0.495\linewidth}
        \centering
        \includegraphics[trim={0 1cm 0 10cm},clip,width=\linewidth]{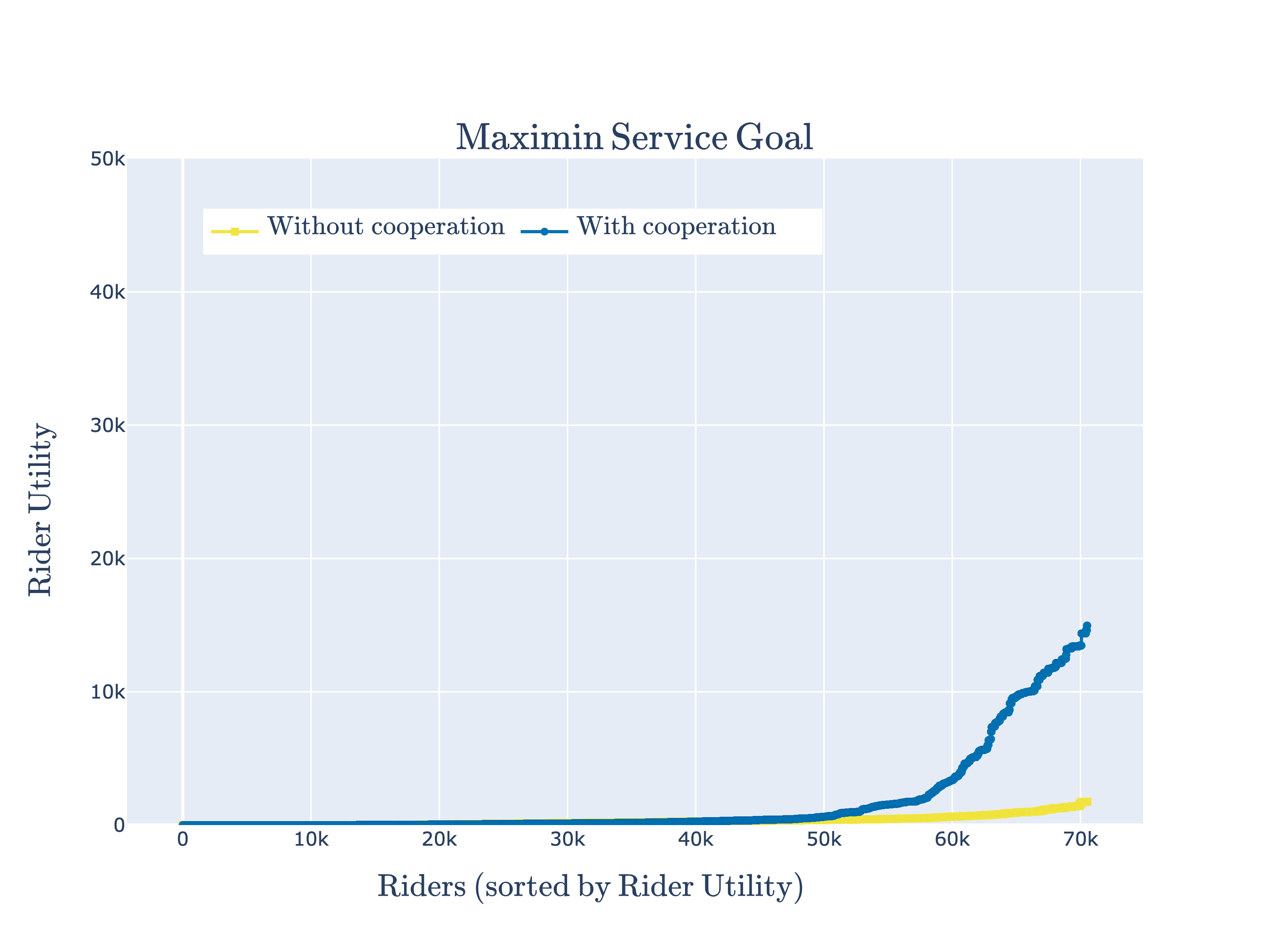}
        \caption{
            Maximin service goal.
        }
        \label{fig: maximin distributive implications}
    \end{subfigure}
    \hfill
    \begin{subfigure}{0.495\linewidth}
        \centering
        \includegraphics[trim={0 1cm 0 10cm},clip,width=\linewidth]{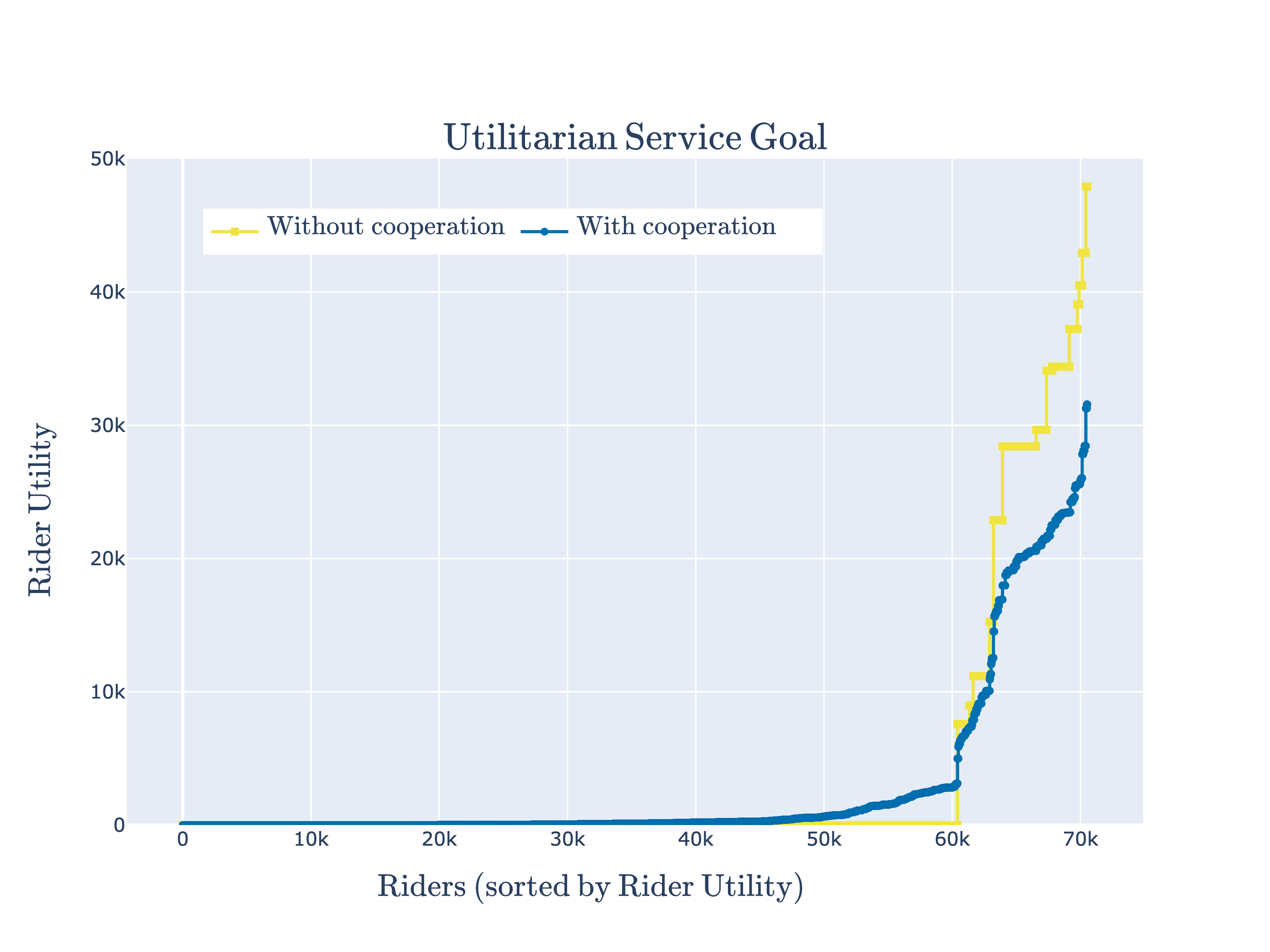}
        \caption{
            Utilitarian service goal.
        }
        \label{fig: utilitarian distributive implications}
    \end{subfigure}
\end{figure}

Figure~\ref{fig: distributive implications} again highlights a main limitation of our assumption of linear utilities \eqref{eq: u_i}: it shows a stark disparity in rider utilities, wherein riders in the top utility quantiles obtain utility that is orders of magnitude greater than the rest.
This is symptomatic of linear utilities' inability to capture diminishing returns that are natural in this and many other applications of public goods.

In Figure~\ref{fig: welfare} we further examine the social welfare implications of cooperation as Algorithm~\ref{alg: algorithm} progresses.
Naturally, the maximin social welfare is higher under the maximin service goal \eqref{eq: maximin} than it is under the utilitarian service goal \eqref{eq: utilitarian}, as shown in Figure~\ref{fig: maximin welfare}.
The counterpart effect on the utilitarian social welfare is shown in Figure~\ref{fig: utilitarian welfare}.
More importantly, the figures show that cooperation has a negative effect on the intended design goal, as shown by the decreasing maximin social welfare in Figure~\ref{fig: maximin welfare} and the decreasing utilitarian social welfare in Figure~\ref{fig: utilitarian welfare}.
However, we stress that this kind of (possibly counterintuitive) tradeoff is necessary for a successful deployment: neglecting cooperation risks coalitional opposition, as in the real-world examples from Section~\ref{sec: motivating example}.
In this sense, \eqref{eq: utilitarian} and \eqref{eq: maximin} without cooperation are not practically attainable benchmarks.
\begin{figure}[ht]
    \centering
    \caption{
        Social welfare implications of approaching a cooperative solution for the maximin and utilitarian service goals.
    }
    \label{fig: welfare}
    \begin{subfigure}{\linewidth}
        \centering
        \includegraphics[trim={0cm 1cm 2cm 2.5cm},clip,width=0.8\linewidth]{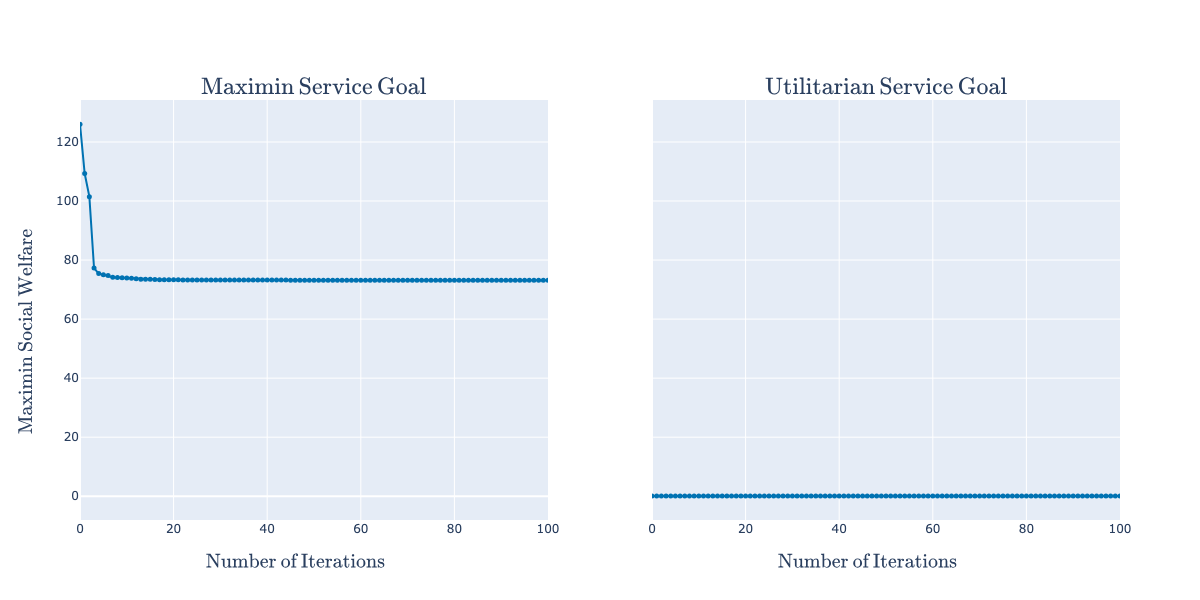}
        \caption{
            Effects on the maximin social welfare.
        }
        \label{fig: maximin welfare}
    \end{subfigure}
    \hfill
    \begin{subfigure}{\linewidth}
        \centering
        \includegraphics[trim={0cm 1cm 2cm 2.5cm},clip,width=0.8\linewidth]{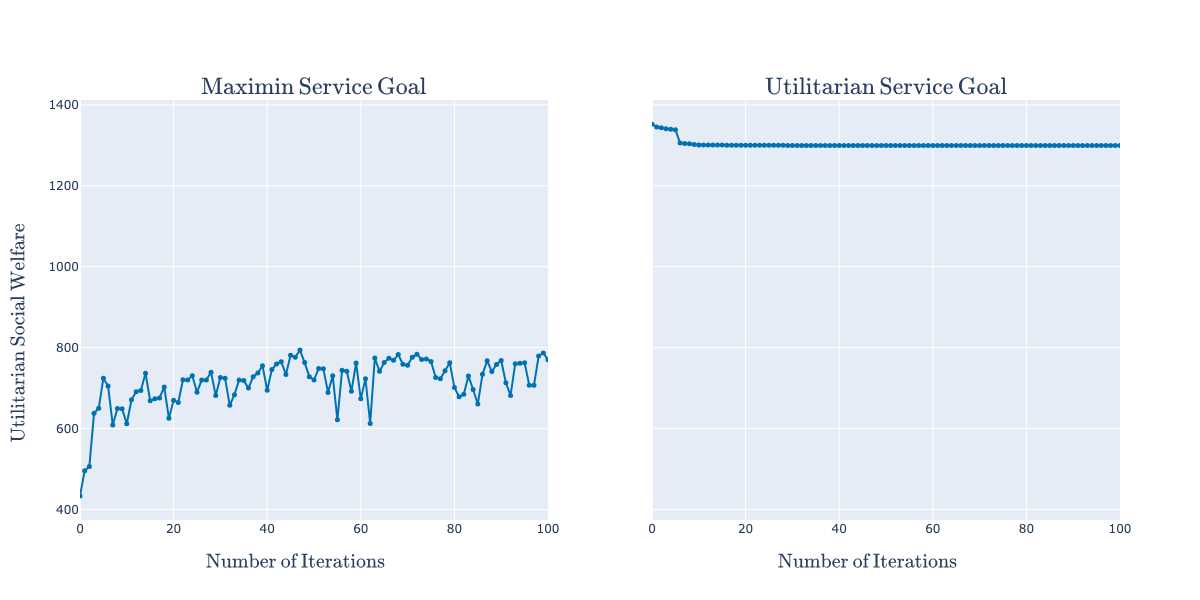}
        \caption{
            Effects on the utilitarian social welfare.
        }
        \label{fig: utilitarian welfare}
    \end{subfigure}
\end{figure}

In Figure~\ref{fig: eps} we consider whether the final service plans after $100$ iterations of our implementation of Algorithm~\ref{alg: algorithm} result in core utilities.
While the objective values in Figures~\ref{fig: maximin welfare} and \ref{fig: utilitarian welfare} appear to become static after this many iterations, a closer look into the multiplicative least objections $\epsilon(u^*)$ observed throughout the execution of the algorithm reveals this to be much more subtle.
For both the maximin and utilitarian service goals, the largest least objections obtained from \eqref{eq: membership} decrease sharply for the first $20$-$30$ iterations of the algorithm.
However, the solutions in subsequent iterations are not formally in the core, as the solver continues to find blocking coalitions with multiplicative objections in the neighborhood of $1.5$ for the maximin service goal and $2$ for the utilitarian service goal.
At this stage, the effectiveness of new cuts appears to stall.

\begin{figure}[ht]
    \centering
    \caption{
        Multiplicative least objections $\epsilon(u^*)$ observed throughout the execution of Algorithm~\ref{alg: algorithm}.
    }
    \label{fig: eps}
    \includegraphics[trim={0cm 1cm 2cm 2.5cm},clip,width=0.8\linewidth]{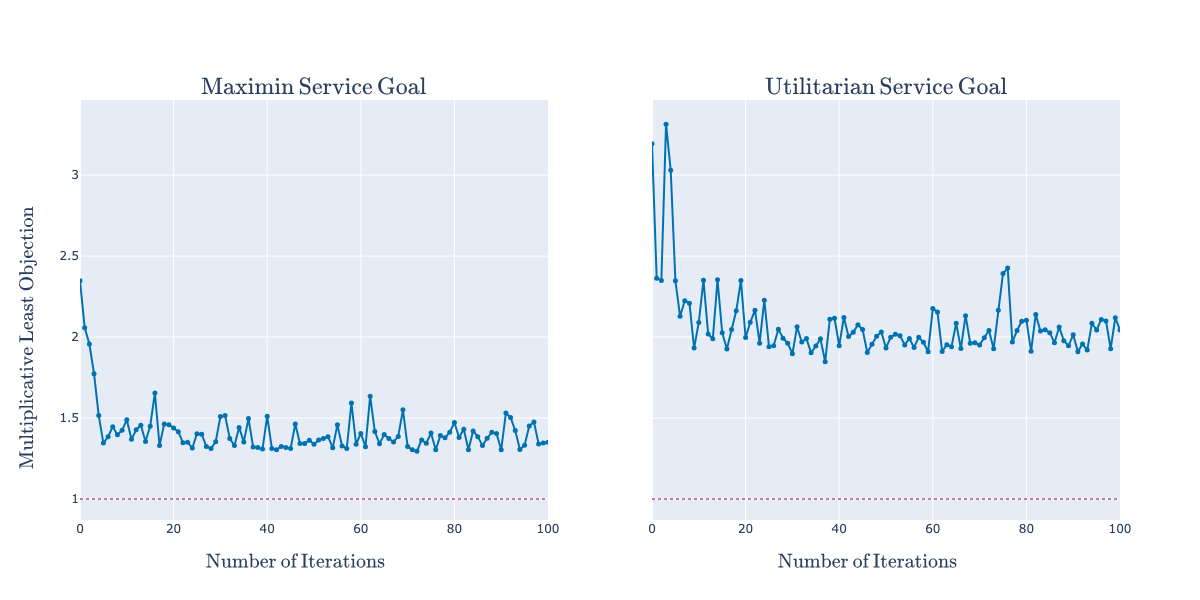}
\end{figure}
Computational convergence difficulties such as those in Figure~\ref{fig: eps} arise in pure cutting plane algorithms as cuts become increasingly parallel and numerically unstable.
In Figure~\ref{fig: kappa} in the Appendix we show that this is indeed the case for our implementation of Algorithm~\ref{alg: algorithm}, as reflected by the condition numbers of the basis matrices observed throughout its execution.
Since these matrices play a fundamental role in the generation of intersection cuts \eqref{eq: zeta_r}, these numerical challenges affect not only the linear programming solver, but crucially our practical ability to generate new cuts.

Lastly, we make some qualitative observations that link the data used to computational performance and the choice of service goal.
In our experiments, and as exemplified in Figures~\ref{fig: maximin iteration service plan} and \ref{fig: utilitarian iteration service plan}, the blocking coalitions under the maximin goal tended to involve northern community areas whereas those under the utilitarian goal tended to involve southern community areas.
This reflects the data used: as shown in Figure~\ref{fig: budgets}, there is a much greater concentration of ride-hailing trips in the north than in the south.
In turn, this reflects pre-existing spatial disparities in Cook County, as ride-hailing users tend to be wealthy and young \citep{young2019and}.
Thus, the utilitarian goal tends to favor the largest concentrations of wealth in the north at the expense of poorer areas in the south; the converse effect holds under the maximin goal.
And as suggested by Figure~\ref{fig: eps}, the areas in the south have relatively more to lose from an utilitarian plan than the areas in the north have to lose from a maximin plan: the former achieve multiplicative least objections around $2$ whereas the latter achieve multiplicative least objections around $1.5$.
Meanwhile, as suggested by Figure~\ref{fig: eps mip gap} in the Appendix, those who benefit from a maximin plan appear to have a harder (computational) time forming blocking coalitions against an utilitarian plan compared to the other way around: the least objections obtained under the maximin goal are optimal, whereas the least objections obtained under the utilitarian goal have an optimality gap that fluctuates between $0-30\%$ at timeout.
These observations point to broader questions about power dynamics in public decision-making environments: which constituent groups tend to benefit more in both absolute and relative terms, how readily can different constituent groups get organized to form blocking coalitions, and how the answers to these questions depend on the public decision-maker's design goal.

\section{Conclusion}
\label{sec: conclusion}

In this work we introduced NTU LP games, which combine the game-theoretic tensions inherent in public decision-making with the modeling flexibility of linear programming.
We derived structural properties regarding the non-emptiness, representability and complexity of the core, a solution concept that models the viability of cooperation.
We also developed and implemented a cutting plane algorithm to optimize linear functions over the core, and illustrated the potentially counterintuitive consequences of cooperation through a data-driven application in public transit.

Cooperation is necessary for any successful real-world implementation of public decision-making: neglecting it risks failures such as those highlighted in Section~\ref{sec: motivating example}.
To bridge this gap, we argue that standard frameworks of mathematical modeling and optimization must be extended to account for the second-order, game-theoretic effects that stem from public reception.
Understanding these elements is crucial to achieving cooperative outcomes that mitigate adverse consequences\textemdash we treat this work as an initial step in this direction.

Our results motivate several future avenues of research.
The main algorithmic challenges are efficient methods to solve and/or verify optimality for the membership test \eqref{eq: membership}, as well as scaling up Algorithm~\ref{alg: algorithm}.
Given the numerical challenges associated with pure cutting plane methods \citep{zanette2011lexicography}, fundamentally different algorithms such as interior point or polyhedral walk-based methods could be viable alternatives.
For any of these approaches, it would be desirable to have simple mechanisms to produce a core allocation that can serve as a suboptimal but feasible start to an optimization procedure over the core (e.g., a successful analogue to Example~\ref{example: synergy}).
Here we draw a connection to the computation of Lindahl equilibria in the context of participatory budgeting problems: for certain instances, these lie in the core \citep{foley1970lindahl} and can be computed efficiently \citep{fain2016core,kroer2025computing}.
Participatory budgeting problems can be cast as a special class of NTU LP games with a single resource, an all-ones production matrix, non-negative resource endowments, and concave utilities.
Therefore, it is natural to ask for core-stable, efficiently-computable generalizations of Lindahl equilibria in the full generality of NTU LP.

It would be similarly important to extend the framework of NTU production games to more expressive classes of mathematical programs.
For example, much of the work on the utility derived from transportation accessibility emphasizes the notion of ``sufficient accessibility'' \citep{martens2012justice}, which suggests the presence of diminishing returns that cannot be captured by linear utilities, as discussed in Section~\ref{sec: operational and distributive implications}.
Moreover, many design problems in public transit systems involve non-convex structures or discrete design variables captured via mixed-integer programming, such as line planning~\citep{borndorfer2007column}.
Therefore, NTU production games with mixed-integer variables and/or with non-linear valuations are natural extensions.
This kind of game-theoretic framework could also be instrumental to deriving domain-specific policy insights, such as the interplay between the cooperative and equity implications of flat versus differentiated fare schemes in public transit systems~\citep{cervero1981flat,borndorfer2015fair}.

Lastly, we recall Algorithm~\ref{alg: algorithm} may conversely certify the emptiness of the core by producing an infeasible system of linear inequalities (ref. Example~\ref{ex: mis}).
From a conceptual point of view, and barring inherent complexity challenges \citep{amaldi2003maximum}, the resulting irreducible infeasible subsystems may direct attention to the fundamental sources of conflict and therefore carve a path forward to their creative resolution.
Specifically, the proposed theoretical framework assumes a fixed set of players, goods, and resources.
However, in practice these elements may be flexible as part of an integrative negotiation process: new players may be brought in, alternative goods may be proposed, and latent resources may be pointed out.
In this sense, Algorithm~\ref{alg: algorithm} may help diagnose any sets of conflicting elements and in turn help derive interventions to resolve them.

\section*{Code and Data Disclosure}\label{sec:Code and Data Disclosure}
Our source code can be found online at \texttt{github.com/jcmartinezmori/ntulp\_v2}.

\section*{Acknowledgments}
We would like to thank the editorial team and the anonymous referees, whose suggestions helped us significantly improve our manuscript.
J.C.\ Mart\'inez Mori is supported by Schmidt Science Fellows, in partnership with the Rhodes Trust.
Part of this research was performed while J.C.\ Mart\'inez Mori was affiliated with the Georgia Institute of Technology as part of the President's Postdoctoral Fellowship Program.
Part of this research was performed while J.C.\ Mart\'inez Mori was visiting the Mathematical Sciences Research Institute (MSRI), now the Simons Laufer Mathematical Sciences Institute (SLMath), which is supported by NSF Grant DMS-192893.

\bibliographystyle{plainnat}
\bibliography{bib.bib}

\begin{thebibliography}{66}
\providecommand{\natexlab}[1]{#1}
\providecommand{\url}[1]{\texttt{#1}}
\expandafter\ifx\csname urlstyle\endcsname\relax
  \providecommand{\doi}[1]{doi: #1}\else
  \providecommand{\doi}{doi: \begingroup \urlstyle{rm}\Url}\fi

\bibitem[Accuardi(2019)]{transit2019derailed}
Zak Accuardi.
\newblock Derailed: How nashville’s ambitious transit plan crashed at the polls\textemdash and what other cities can learn from it.
\newblock Technical report, TransitCenter, New York, NY, 2019.

\bibitem[Amaldi et~al.(2003)Amaldi, Pfetsch, and Trotter]{amaldi2003maximum}
Edoardo Amaldi, Marc~E. Pfetsch, and Leslie~E. Trotter, Jr.
\newblock On the maximum feasible subsystem problem, {IIS}s and {IIS}-hypergraphs.
\newblock \emph{Math. Program.}, 95\penalty0 (3):\penalty0 533--554, 2003.
\newblock ISSN 0025-5610,1436-4646.
\newblock \doi{10.1007/s10107-002-0363-5}.
\newblock URL \url{https://doi.org/10.1007/s10107-002-0363-5}.

\bibitem[Aumann and Peleg(1960)]{aumann1960von}
R.~J. Aumann and B.~Peleg.
\newblock Von {N}eumann-{M}orgenstern solutions to co-operative games without side payments.
\newblock \emph{Bull. Amer. Math. Soc.}, 66:\penalty0 173--179, 1960.
\newblock ISSN 0002-9904.
\newblock \doi{10.1090/S0002-9904-1960-10418-1}.
\newblock URL \url{https://doi.org/10.1090/S0002-9904-1960-10418-1}.

\bibitem[Balas(1971)]{balas1971intersection}
Egon Balas.
\newblock Intersection cuts---a new type of cutting planes for integer programming.
\newblock \emph{Operations Res.}, 19:\penalty0 19--39, 1971.
\newblock ISSN 0030-364X,1526-5463.
\newblock \doi{10.1287/opre.19.1.19}.
\newblock URL \url{https://doi.org/10.1287/opre.19.1.19}.

\bibitem[Billera(1970)]{billera1970some}
Louis~J. Billera.
\newblock Some theorems on the core of an {$n$}-person game without side-payments.
\newblock \emph{SIAM J. Appl. Math.}, 18:\penalty0 567--579, 1970.
\newblock ISSN 0036-1399.
\newblock \doi{10.1137/0118049}.
\newblock URL \url{https://doi.org/10.1137/0118049}.

\bibitem[Boeing(2017)]{boeing2017osmnx}
Geoff Boeing.
\newblock Osmnx: New methods for acquiring, constructing, analyzing, and visualizing complex street networks.
\newblock \emph{Computers, Environment and Urban Systems}, 65:\penalty0 126--139, 2017.
\newblock \doi{10.1016/j.compenvurbsys.2017.05.004}.
\newblock URL \url{https://doi.org/10.1016/j.compenvurbsys.2017.05.004}.

\bibitem[Bondareva(1963)]{bondareva1963some}
Olga~N. Bondareva.
\newblock Some applications of the methods of linear programming to the theory of cooperative games.
\newblock \emph{Problemy Kibernet.}, 10:\penalty0 119--139, 1963.

\bibitem[Bornd\"orfer and Hoang(2015)]{borndorfer2015fair}
Ralf Bornd\"orfer and Nam-D\~ung Hoang.
\newblock Fair ticket pricing in public transport as a constrained cost allocation game.
\newblock \emph{Ann. Oper. Res.}, 226:\penalty0 51--68, 2015.
\newblock ISSN 0254-5330,1572-9338.
\newblock \doi{10.1007/s10479-014-1698-z}.
\newblock URL \url{https://doi.org/10.1007/s10479-014-1698-z}.

\bibitem[Bornd{\"o}rfer et~al.(2007)Bornd{\"o}rfer, Gr{\"o}tschel, and Pfetsch]{borndorfer2007column}
Ralf Bornd{\"o}rfer, Martin Gr{\"o}tschel, and Marc~E Pfetsch.
\newblock A column-generation approach to line planning in public transport.
\newblock \emph{Transportation Science}, 41\penalty0 (1):\penalty0 123--132, 2007.
\newblock \doi{10.1287/trsc.1060.0161}.
\newblock URL \url{https://doi.org/10.1287/trsc.1060.0161}.

\bibitem[Caddell(2005)]{caddell2005plans}
Marlin Caddell.
\newblock Plans in the works for campus transit system, 2005.
\newblock URL \url{https://now.dirxion.com/Crimson_White/library/Crimson_White_11_4_2005.pdf#page=1&zoom=100}.

\bibitem[Cervero(1981)]{cervero1981flat}
Robert Cervero.
\newblock Flat versus differentiated transit pricing: what's a fair fare?
\newblock \emph{Transportation}, 10\penalty0 (3):\penalty0 211--232, 1981.
\newblock \doi{10.1007/BF00148459}.
\newblock URL \url{https://doi.org/10.1007/BF00148459}.

\bibitem[Chen and Hooker(2023)]{chen2023guide}
Violet~Xinying Chen and J.~N. Hooker.
\newblock A guide to formulating fairness in an optimization model.
\newblock \emph{Ann. Oper. Res.}, 326\penalty0 (1):\penalty0 581--619, 2023.
\newblock ISSN 0254-5330,1572-9338.
\newblock \doi{10.1007/s10479-023-05264-y}.
\newblock URL \url{https://doi.org/10.1007/s10479-023-05264-y}.

\bibitem[Chen et~al.(2009)Chen, Deng, and Teng]{chen2009settling}
Xi~Chen, Xiaotie Deng, and Shang-Hua Teng.
\newblock Settling the complexity of computing two-player {N}ash equilibria.
\newblock \emph{J. ACM}, 56\penalty0 (3):\penalty0 Art. 14, 57, 2009.
\newblock ISSN 0004-5411,1557-735X.
\newblock \doi{10.1145/1516512.1516516}.
\newblock URL \url{https://doi.org/10.1145/1516512.1516516}.

\bibitem[Chen and Zhang(2009)]{chen2009stoch}
Xin Chen and Jiawei Zhang.
\newblock {A Stochastic Programming Duality Approach to Inventory Centralization Games}.
\newblock \emph{Operations Research}, 57\penalty0 (4):\penalty0 840--851, 2009.

\bibitem[Chen and Zhang(2016)]{chen2016duality}
Xin Chen and Jiawei Zhang.
\newblock Duality approaches to economic lot-sizing games.
\newblock \emph{Production and Operations Management}, 25\penalty0 (7):\penalty0 1203--1215, 2016.
\newblock \doi{10.1111/poms.12542}.

\bibitem[{Chicago Data Portal}(2025)]{cdp}
{Chicago Data Portal}.
\newblock Transportation network providers - trips (2023-), 2025.
\newblock URL \url{https://data.cityofchicago.org/Transportation/Transportation-Network-Providers-Trips-2023-/n26f-ihde/about_data}.

\bibitem[Conforti et~al.(2014)Conforti, Cornu\'ejols, and Zambelli]{conforti2014integer}
Michele Conforti, G\'erard Cornu\'ejols, and Giacomo Zambelli.
\newblock \emph{Integer programming}, volume 271 of \emph{Graduate Texts in Mathematics}.
\newblock Springer, Cham, Switzerland, 2014.
\newblock ISBN 978-3-319-11007-3; 978-3-319-11008-0.
\newblock \doi{10.1007/978-3-319-11008-0}.
\newblock URL \url{https://doi.org/10.1007/978-3-319-11008-0}.

\bibitem[Daskalakis et~al.(2009)Daskalakis, Goldberg, and Papadimitriou]{daskalakis2009complexity}
Constantinos Daskalakis, Paul~W. Goldberg, and Christos~H. Papadimitriou.
\newblock The complexity of computing a {N}ash equilibrium.
\newblock \emph{SIAM J. Comput.}, 39\penalty0 (1):\penalty0 195--259, 2009.
\newblock ISSN 0097-5397,1095-7111.
\newblock \doi{10.1137/070699652}.
\newblock URL \url{https://doi.org/10.1137/070699652}.

\bibitem[Dougherty(2024)]{dougherty2024tcat}
Matt Dougherty.
\newblock {TCAT}-{C}ornell service agreement strains operations, threatens community routes, 2024.
\newblock URL \url{https://www.ithaca.com/news/ithaca/tcat-cornell-service-agreement-strains-operations-threatens-community-routes/article_f9c7aa00-12e0-11ef-860d-0bceafbe2158.html}.

\bibitem[Ellenberg(2021)]{ellenberg2021geometry}
Jordan~S. Ellenberg.
\newblock Geometry, inference, complexity, and democracy.
\newblock \emph{Bull. Amer. Math. Soc. (N.S.)}, 58\penalty0 (1):\penalty0 57--77, 2021.
\newblock ISSN 0273-0979,1088-9485.
\newblock \doi{10.1090/bull/1708}.
\newblock URL \url{https://doi.org/10.1090/bull/1708}.

\bibitem[Fain et~al.(2016)Fain, Goel, and Munagala]{fain2016core}
Brandon Fain, Ashish Goel, and Kamesh Munagala.
\newblock The core of the participatory budgeting problem.
\newblock In \emph{International Conference on Web and Internet Economics}, pages 384--399. Springer, 2016.

\bibitem[Fang et~al.(2001)Fang, Zhu, Cai, and Deng]{fang2001membership}
Qizhi Fang, Shanfeng Zhu, Maocheng Cai, and Xiaotie Deng.
\newblock Membership for core of {LP} games and other games.
\newblock In \emph{Computing and combinatorics ({G}uilin, 2001)}, volume 2108 of \emph{Lecture Notes in Comput. Sci.}, pages 247--256. Springer, Berlin, 2001.
\newblock ISBN 3-540-42494-6.
\newblock \doi{10.1007/3-540-44679-6\_27}.
\newblock URL \url{https://doi.org/10.1007/3-540-44679-6_27}.

\bibitem[Foley(1970)]{foley1970lindahl}
Duncan~K Foley.
\newblock Lindahl's solution and the core of an economy with public goods.
\newblock \emph{Econometrica: Journal of the Econometric Society}, pages 66--72, 1970.

\bibitem[Gillies(1959)]{gillies59}
Donald~B. Gillies.
\newblock Solutions to general non-zero-sum games.
\newblock \emph{Contributions to the Theory of Games}, 4\penalty0 (40):\penalty0 47--85, 1959.

\bibitem[Goemans and Skutella(2004)]{goemans2004cooperative}
Michel~A. Goemans and Martin Skutella.
\newblock Cooperative facility location games.
\newblock \emph{J. Algorithms}, 50\penalty0 (2):\penalty0 194--214, 2004.
\newblock ISSN 0196-6774.
\newblock \doi{10.1016/S0196-6774(03)00098-1}.
\newblock URL \url{https://doi.org/10.1016/S0196-6774(03)00098-1}.
\newblock SODA 2000 special issue.

\bibitem[Gui and Ergun(2008)]{gui2008dual}
Luyi Gui and \"Ozlem Ergun.
\newblock Dual payoffs, core and a collaboration mechanism based on capacity exchange prices in multicommodity flow games.
\newblock In Christos Papadimitriou and Shuzhong Zhang, editors, \emph{Internet and Network Economics (WINE 2008)}, volume 5385 of \emph{Lecture Notes in Computer Science}, pages 61--69. Springer, Berlin Heidelberg, 2008.
\newblock \doi{10.1007/978-3-540-92185-1_15}.

\bibitem[{Gurobi Optimization, LLC}(2025)]{gurobi2025}
{Gurobi Optimization, LLC}.
\newblock \emph{Gurobi Optimizer Reference Manual}.
\newblock Gurobi Optimization, LLC, 2025.
\newblock URL \url{https://docs.gurobi.com/_/downloads/optimizer/en/12.0/pdf/}.
\newblock Chapter 7: Guidelines for Numerical Issues.

\bibitem[Howell and Timberlake(2014)]{howell2014racial}
Aaron~J. Howell and Jeffrey~M. Timberlake.
\newblock Racial and ethnic trends in the suburbanization of poverty in us metropolitan areas, 1980--2010.
\newblock \emph{Journal of Urban Affairs}, 36\penalty0 (1):\penalty0 79--98, 2014.
\newblock \doi{10.1111/juaf.12030}.
\newblock URL \url{https://doi.org/10.1111/juaf.12030}.

\bibitem[Jeroslow and Lowe(1984)]{jeroslow1984modelling}
R.~G. Jeroslow and J.~K. Lowe.
\newblock Modelling with integer variables.
\newblock In Bernhard Korte and Klaus Ritter, editors, \emph{Mathematical Programming at Oberwolfach II}, volume~22 of \emph{Mathematical Programming Studies}, pages 167--184. Springer, Berlin Heidelberg, 1984.
\newblock \doi{10.1007/bfb0121015}.
\newblock URL \url{https://doi.org/10.1007/bfb0121015}.

\bibitem[Karp(1972)]{karp1972reducibility}
Richard~M. Karp.
\newblock Reducibility among combinatorial problems.
\newblock In \emph{Complexity of computer computations ({P}roc. {S}ympos., {IBM} {T}homas {J}. {W}atson {R}es. {C}enter, {Y}orktown {H}eights, {N}.{Y}., 1972)}, The IBM Research Symposia Series, pages 85--103. Plenum, New York-London, 1972.

\bibitem[Kintali et~al.(2013)Kintali, Poplawski, Rajaraman, Sundaram, and Teng]{kintali2013reducibility}
Shiva Kintali, Laura~J. Poplawski, Rajmohan Rajaraman, Ravi Sundaram, and Shang-Hua Teng.
\newblock Reducibility among fractional stability problems.
\newblock \emph{SIAM J. Comput.}, 42\penalty0 (6):\penalty0 2063--2113, 2013.
\newblock ISSN 0097-5397,1095-7111.
\newblock \doi{10.1137/120874655}.
\newblock URL \url{https://doi.org/10.1137/120874655}.

\bibitem[Kneebone and Berube(2023)]{kneebone2023post}
Elizabeth Kneebone and Alan Berube.
\newblock Post-pandemic poverty is rising in america's suburbs, 2023.
\newblock URL \url{https://www.brookings.edu/articles/post-pandemic-poverty-is-rising-in-americas-suburbs/?__cf_chl_tk=BQpKvo2uUTB0Nmg901WpNy6csT4qXuA5iWnbvg0qhEg-1734539872-1.0.1.1-NG4w8QOHr8m_QofvKHbbn2BCySW4ni2CHTfyTClyeuo}.

\bibitem[Kramer(2018)]{kramer2018unaffordable}
Anna Kramer.
\newblock The unaffordable city: Housing and transit in north american cities.
\newblock \emph{Cities}, 83:\penalty0 1--10, 2018.
\newblock \doi{10.1016/j.cities.2018.05.013}.
\newblock URL \url{https://doi.org/10.1016/j.cities.2018.05.013}.

\bibitem[Kroer and Peters(2025)]{kroer2025computing}
Christian Kroer and Dominik Peters.
\newblock Computing lindahl equilibrium for public goods with and without funding caps.
\newblock In \emph{Proceedings of the 26th ACM Conference on Economics and Computation}, pages 129--129, 2025.

\bibitem[Lemke and Howson(1964)]{lemke1964equilibrium}
C.~E. Lemke and J.~T. Howson, Jr.
\newblock Equilibrium points of bimatrix games.
\newblock \emph{J. Soc. Indust. Appl. Math.}, 12:\penalty0 413--423, 1964.
\newblock ISSN 0368-4245.

\bibitem[Leyton-Brown et~al.(2017)Leyton-Brown, Milgrom, and Segal]{leyton2017economics}
Kevin Leyton-Brown, Paul Milgrom, and Ilya Segal.
\newblock Economics and computer science of a radio spectrum reallocation.
\newblock \emph{Proceedings of the National Academy of Sciences}, 114\penalty0 (28):\penalty0 7202--7209, 2017.

\bibitem[Markakis and Saberi(2003)]{markakis2003core}
Evangelos Markakis and Amin Saberi.
\newblock On the core of the multicommodity flow game.
\newblock In \emph{Proceedings of the 4th ACM Conference on Electronic Commerce}, pages 93--97. ACM, New York, NY, 2003.
\newblock ISBN 158113679X.
\newblock \doi{10.1145/779928.779940}.

\bibitem[Martens(2012)]{martens2012justice}
Karel Martens.
\newblock Justice in transport as justice in accessibility: applying walzer’s ``spheres of justice'' to the transport sector.
\newblock \emph{Transportation}, 39:\penalty0 1035--1053, 2012.
\newblock \doi{10.1007/s11116-012-9388-7}.
\newblock URL \url{https://doi.org/10.1007/s11116-012-9388-7}.

\bibitem[McMullen(1970)]{mcmullen1970maximum}
P.~McMullen.
\newblock The maximum numbers of faces of a convex polytope.
\newblock \emph{Mathematika}, 17:\penalty0 179--184, 1970.
\newblock ISSN 0025-5793.
\newblock \doi{10.1112/S0025579300002850}.
\newblock URL \url{https://doi.org/10.1112/S0025579300002850}.

\bibitem[Mock(2023)]{mock2023how}
Brentin Mock.
\newblock How {B}uckhead’s secession from {A}tlanta would destabilize the entire state, 2023.
\newblock URL \url{https://www.bloomberg.com/news/articles/2023-03-02/buckhead-secession-effort-threatens-financial-health-of-atlanta-region}.

\bibitem[Myerson(1980)]{myerson1980conference}
R.~B. Myerson.
\newblock Conference structures and fair allocation rules.
\newblock \emph{Internat. J. Game Theory}, 9\penalty0 (3):\penalty0 169--182, 1980.
\newblock ISSN 0020-7276,1432-1270.
\newblock \doi{10.1007/BF01781371}.
\newblock URL \url{https://doi.org/10.1007/BF01781371}.

\bibitem[Myerson(1977)]{myerson1977graphs}
Roger~B. Myerson.
\newblock Graphs and cooperation in games.
\newblock \emph{Math. Oper. Res.}, 2\penalty0 (3):\penalty0 225--229, 1977.
\newblock ISSN 0364-765X,1526-5471.
\newblock \doi{10.1287/moor.2.3.225}.
\newblock URL \url{https://doi.org/10.1287/moor.2.3.225}.

\bibitem[Nash(1950)]{nash1950equilibrium}
John~F. Nash, Jr.
\newblock Equilibrium points in {$n$}-person games.
\newblock \emph{Proc. Nat. Acad. Sci. U.S.A.}, 36:\penalty0 48--49, 1950.
\newblock ISSN 0027-8424.
\newblock \doi{10.1073/pnas.36.1.48}.
\newblock URL \url{https://doi.org/10.1073/pnas.36.1.48}.

\bibitem[{OpenStreetMap contributors}(2025)]{osm}
{OpenStreetMap contributors}, 2025.
\newblock URL \url{https://www.openstreetmap.org}.

\bibitem[Owen(1975)]{owen1975core}
Guillermo Owen.
\newblock On the core of linear production games.
\newblock \emph{Math. Programming}, 9\penalty0 (3):\penalty0 358--370, 1975.
\newblock ISSN 0025-5610,1436-4646.
\newblock \doi{10.1007/BF01681356}.
\newblock URL \url{https://doi.org/10.1007/BF01681356}.

\bibitem[Papadimitriou(1994)]{papadimitriou1994complexity}
Christos~H. Papadimitriou.
\newblock On the complexity of the parity argument and other inefficient proofs of existence.
\newblock \emph{J. Comput. System Sci.}, 48\penalty0 (3):\penalty0 498--532, 1994.
\newblock ISSN 0022-0000,1090-2724.
\newblock \doi{10.1016/S0022-0000(05)80063-7}.
\newblock URL \url{https://doi.org/10.1016/S0022-0000(05)80063-7}.

\bibitem[Peleg and Sudh\"olter(2007)]{peleg2007introduction}
Bezalel Peleg and Peter Sudh\"olter.
\newblock \emph{Introduction to the theory of cooperative games}, volume~34 of \emph{Theory and Decision Library. Series C: Game Theory, Mathematical Programming and Operations Research}.
\newblock Springer, Berlin Heidelberg, second edition, 2007.
\newblock ISBN 978-3-540-72944-0.

\bibitem[Porembski(2001)]{porembski2001finitely}
Marcus Porembski.
\newblock Finitely convergent cutting planes for concave minimization.
\newblock \emph{J. Global Optim.}, 20\penalty0 (2):\penalty0 113--136, 2001.
\newblock ISSN 0925-5001,1573-2916.
\newblock \doi{10.1023/A:1011240309783}.
\newblock URL \url{https://doi.org/10.1023/A:1011240309783}.

\bibitem[Robinson(2024)]{robinson2024cobb}
Jasmine Robinson.
\newblock Cobb and gwinnett voters reject public transit expansion proposed in ballot measures, 2024.
\newblock URL \url{https://www.wabe.org/cobb-and-gwinnett-voters-reject-public-transit-expansion-proposed-in-ballot-measures/}.

\bibitem[Rojas(2024)]{rojas2024louisiana}
Rick Rojas.
\newblock Louisiana will get a new city after a yearslong court battle, 2024.
\newblock URL \url{https://www.nytimes.com/2024/04/28/us/louisiana-st-george.html?smid=url-share}.

\bibitem[Scarf(1967)]{scarf1967core}
Herbert~E. Scarf.
\newblock The core of an {$N$} person game.
\newblock \emph{Econometrica}, 35:\penalty0 50--69, 1967.
\newblock ISSN 0012-9682,1468-0262.
\newblock \doi{10.2307/1909383}.
\newblock URL \url{https://doi.org/10.2307/1909383}.

\bibitem[Schulz and Uhan(2010)]{schulz2010sharing}
Andreas~S. Schulz and Nelson~A. Uhan.
\newblock Sharing supermodular costs.
\newblock \emph{Oper. Res.}, 58\penalty0 (4):\penalty0 1051--1056, 2010.
\newblock ISSN 0030-364X,1526-5463.
\newblock \doi{10.1287/opre.1100.0841}.
\newblock URL \url{https://doi.org/10.1287/opre.1100.0841}.

\bibitem[Shapley(1967)]{shapley1967balanced}
Lloyd~S. Shapley.
\newblock On balanced sets and cores.
\newblock \emph{Naval Res. Logist.}, 14\penalty0 (4):\penalty0 453--460, 1967.
\newblock \doi{10.1002/nav.3800140404}.
\newblock URL \url{https://doi.org/10.1002/nav.3800140404}.

\bibitem[Shapley(1973)]{shapley1973balanced}
Lloyd~S. Shapley.
\newblock On balanced games without side payments.
\newblock In T.~C. Hu and Stephen~M. Robinson, editors, \emph{Mathematical Programming}, pages 261--290. Academic Press, 1973.
\newblock \doi{10.1016/B978-0-12-358350-5.50012-9}.
\newblock URL \url{https://doi.org/10.1016/B978-0-12-358350-5.50012-9}.

\bibitem[Shapley and Shubik(1966)]{shapley1966quasi}
Lloyd~S Shapley and Martin Shubik.
\newblock Quasi-cores in a monetary economy with nonconvex preferences.
\newblock \emph{Econometrica: Journal of the Econometric Society}, pages 805--827, 1966.

\bibitem[Sperner(1928)]{sperner1928neuer}
E.~Sperner.
\newblock Neuer beweis f\"ur die invarianz der dimensionszahl und des gebietes.
\newblock \emph{Abh. Math. Sem. Univ. Hamburg}, 6\penalty0 (1):\penalty0 265--272, 1928.
\newblock ISSN 0025-5858,1865-8784.
\newblock \doi{10.1007/BF02940617}.
\newblock URL \url{https://doi.org/10.1007/BF02940617}.

\bibitem[Toriello and Uhan(2013)]{toriello2013uhan}
Alejandro Toriello and Nelson~A. Uhan.
\newblock Technical note---on traveling salesman games with asymmetric costs.
\newblock \emph{Oper. Res.}, 61\penalty0 (6):\penalty0 1429--1434, 2013.
\newblock ISSN 0030-364X,1526-5463.
\newblock \doi{10.1287/opre.2013.1225}.
\newblock URL \url{https://doi.org/10.1287/opre.2013.1225}.

\bibitem[Toriello and Uhan(2014)]{toriello2014dynamic}
Alejandro Toriello and Nelson~A. Uhan.
\newblock Dynamic cost allocation for economic lot sizing games.
\newblock \emph{Oper. Res. Lett.}, 42\penalty0 (1):\penalty0 82--84, 2014.
\newblock ISSN 0167-6377,1872-7468.
\newblock \doi{10.1016/j.orl.2013.12.005}.
\newblock URL \url{https://doi.org/10.1016/j.orl.2013.12.005}.

\bibitem[Toriello and Uhan(2017)]{toriello2017dynamic}
Alejandro Toriello and Nelson~A. Uhan.
\newblock Dynamic linear programming games with risk-averse players.
\newblock \emph{Math. Program.}, 163\penalty0 (1-2):\penalty0 25--56, 2017.
\newblock ISSN 0025-5610,1436-4646.
\newblock \doi{10.1007/s10107-016-1054-y}.
\newblock URL \url{https://doi.org/10.1007/s10107-016-1054-y}.

\bibitem[Tuy(1964)]{tuy1964concave}
Hoang Tuy.
\newblock Concave programming under linear constraints.
\newblock \emph{Doklady Akademii Nauk}, 5:\penalty0 1437--1440, 1964.

\bibitem[Vielma(2015)]{vielma2015mixed}
Juan~Pablo Vielma.
\newblock Mixed integer linear programming formulation techniques.
\newblock \emph{SIAM Rev.}, 57\penalty0 (1):\penalty0 3--57, 2015.
\newblock ISSN 0036-1445,1095-7200.
\newblock \doi{10.1137/130915303}.
\newblock URL \url{https://doi.org/10.1137/130915303}.

\bibitem[von Neumann and Morgenstern(1944)]{vonneumann1944theory}
John von Neumann and Oskar Morgenstern.
\newblock \emph{Theory of {G}ames and {E}conomic {B}ehavior}.
\newblock Princeton University Press, Princeton, NJ, 1944.

\bibitem[Walker(2024)]{walker2024human}
Jarrett Walker.
\newblock \emph{Human Transit: How Clearer Thinking about Public Transit Can Enrich Our Communities and Our Lives}.
\newblock Island Press, Washington, DC, 2 edition, 2024.

\bibitem[Walzer(1983)]{walzer1983spheres}
Michael Walzer.
\newblock \emph{Spheres of Justice: A Defense of Pluralism and Equality}.
\newblock Basic Books, New York, NY, 1983.

\bibitem[Young and Farber(2019)]{young2019and}
Mischa Young and Steven Farber.
\newblock The who, why, and when of uber and other ride-hailing trips: An examination of a large sample household travel survey.
\newblock \emph{Transportation Research Part A: Policy and Practice}, 119:\penalty0 383--392, 2019.

\bibitem[Zanette et~al.(2011)Zanette, Fischetti, and Balas]{zanette2011lexicography}
Arrigo Zanette, Matteo Fischetti, and Egon Balas.
\newblock Lexicography and degeneracy: can a pure cutting plane algorithm work?
\newblock \emph{Math. Program.}, 130\penalty0 (1):\penalty0 153--176, 2011.
\newblock ISSN 0025-5610,1436-4646.
\newblock \doi{10.1007/s10107-009-0335-0}.
\newblock URL \url{https://doi.org/10.1007/s10107-009-0335-0}.

\end{thebibliography}

\newpage
\section*{Appendix: Additional Proofs and Figures}

\subsection*{Proof of Lemma~\ref{lemma: difference}}
\begin{proof}{Proof}
Consider any $x \in P$.
First, suppose $\interior{Q} \neq \emptyset$.
Then, there are two possibilities.
If $x \in \interior{Q}$, then $a_i^T x < b_i$ for all $i \in [m]$.
Therefore, there is no $i \in [m]$ for which $x \in P \cap \{x' \in \mathbb{R}^n : a_i^Tx' \geq b_i\}$, so that $x$ does not appear in the set union.
On the other hand, if $x \notin \interior{Q}$, then $a_i^T x \geq b_i$ for some $i \in [m]$.
Therefore, there is some $i \in [m]$ for which $x \in P \cap \{x' \in \mathbb{R}^n : a_i^Tx' \geq b_i\}$, so that $x$ appears in the set union.
Conversely, suppose $\interior{Q} = \emptyset$.
Then, $Q$ is not full-dimensional, so that $\bigcup_{i = 1}^m P \cap \{x \in \mathbb{R}^n : a_i^Tx \geq b_i\} = P \cap \left(\bigcup_{i = 1}^m \{x \in \mathbb{R}^n : a_i^Tx \geq b_i\}  \right) = P \cap \mathbb{R}^n = P = P \setminus \emptyset = P \setminus \interior{Q}$.
\end{proof}

\subsection*{Proof of Lemma~\ref{lemma: hardness structural lemma 1}}
\begin{proof}{Proof}
The diagonal structure of $v_j^t$ for $t, j \in T$ together with the fact that $m \cdot v_{m+1}^t =\frac{m}{3n+m} \leq 1$ for all $t \in T$ imply that $S \cap H \neq \emptyset$.
Therefore, consider any $h \in S \cap H$.
Note that
\begin{equation*}
    \underbrace{\frac{1}{4n-1}}_{v_t^h \text{ if } h \in t} > \underbrace{\frac{1}{4n} \left(1 - \frac{1}{2(n-1)(4n-1)} \right)}_{v_t^h \text{ if } h \notin t} > \underbrace{\frac{1}{3n+m}}_{v_{m+1}^h}.
\end{equation*}
Therefore, for $h$ to form part of a blocking coalition against $u^*$, it must be that $b(S) > 4n - 1$.
Since $b(S) = |S|$, and since $|S|$ and $4n - 1$ are integers, this implies that $|S| \geq 4n$.
Together with the facts that $|S| = |S \cap H| + |S \cap T|$ and $|S \cap H| \leq 3n$, this further implies that 
\begin{equation}
\label{eq: t' lb}
   |S \cap T| \geq n.
\end{equation}
Now, consider the collective utility of $S \cap T$.
We distinguish two mutually exclusive possibilities:
\begin{itemize}
    \item
    Collectively, the members of $S \cap T$ gain strictly more utility if $b(S)$ is allocated to good $m + 1$ than they do if $b(S)$ is allocated to goods $S \cap T$.
    That is, $b(S) \cdot \frac{1}{3n + m} \cdot |S \cap T| > b(S) \cdot \frac{n}{4n-1}$.
    Then, for the members of $S \cap T$ to form part of a blocking coalition against $u^*$, it must be that
    $b(S) \cdot \frac{1}{3n + m} \cdot |S \cap T| > |S \cap T|$.
    This implies $b(S) > 3n + m$, so that $|S| > 3n + m$.
    This contradicts the fact that $|S| \leq |N| = 3n + m$, so this possibility cannot be.
    \item
    Collectively, the members of $S \cap T$ gain no more utility if $b(S)$ is allocated to good $m + 1$ than they do if $b(S)$ is allocated to goods $S \cap T$.
    That is, $b(S) \cdot \frac{1}{3n + m} \cdot |S \cap T| \leq b(S) \cdot \frac{n}{4n-1}$.
    Then, for the members of $S \cap T$ to form part of a blocking coalition against $u^*$, it must be that $b(S) \cdot \frac{n}{4n-1}  > |S \cap T|$, which implies $|S| \cdot \frac{n}{4n-1}  > |S \cap T|$.
    Together with the facts that $|S| = |S \cap H| + |S \cap T|$ and $|S \cap H| \leq 3n$, this further implies that $\left(3n + |S \cap T|\right) \cdot \frac{n}{4n-1}  > |S \cap T|$, which upon rearranging becomes $n + \frac{n}{3n - 1} > |S \cap T|$.
    Since $n$ and $|S \cap T|$ are integers and $\frac{n}{3n - 1} < 1$, this in turn implies that $n \geq |S \cap T|$.
\end{itemize}
Together with \eqref{eq: t' lb}, this shows that $|S \cap T| = n$.
Since $|S \cap H| \leq 3n$ and we require $|S| \geq 4n$, this further implies that $|S \cap H| = 3n$.
\end{proof}

\subsection*{Proof of Lemma~\ref{lemma: hardness structural lemma 2}}
\begin{proof}{Proof}
Lemma~\ref{lemma: hardness structural lemma 1} implies that $|S| = 4n$ and $|S \cap T| = n$.
In particular, $b(S) = |S| = 4n$.
By way of contradiction, suppose there exists some member $h \in S \cap H$ such that $h \notin t$ for all $t \in S \cap T$.
If all of $b(S)$ were allocated among goods $S \cap T$, then the utility of $h$ would be given by
\begin{align}
\label{eq: nodes baseline}
    \underbrace{4n}_{b(S)} \cdot & \underbrace{\left(\frac{1}{4n} \left(1 - \frac{1}{2(n-1)(4n-1)} \right)\right)}_{v_t^h \text{ if } h \notin t} \nonumber \\
    &= 1 - \frac{1}{2(n-1)(4n-1)}.
\end{align}
However, since $h \in S \cap H$ forms part of a blocking coalition against $u^*$, it must be that the utility of $h$ exceeds \eqref{eq: nodes baseline} by strictly more than
\begin{equation}
\label{eq: nodes threshold}
    \frac{1}{2(n-1)(4n-1)}
\end{equation}
units of utility.
Next, note that
\begin{equation*}
    \underbrace{\frac{1}{4n-1}}_{v_t^h \text{ if } h \in t}  > \underbrace{\left(\frac{1}{4n} \left(1 - \frac{1}{2(n-1)(4n-1)} \right)\right)}_{v_t^h \text{ if } h \notin t} > \underbrace{\frac{1}{3n + m}}_{v_{m+1}^h}.
\end{equation*}
Therefore, given the assumption that $h \notin t$ for all $t \in S \cap T$, it must be that part of $b(S)$ is allocated to goods $t \in T \setminus S$ with $h \in t$ rather than goods $S \cap T$.
For each such unit of budget, $h$ gains
\begin{align}
\label{eq: nodes rate}
    \underbrace{\frac{1}{4n-1}}_{v_t^h \text{ if } h \in t} - & \underbrace{\left(\frac{1}{4n} \left(1 - \frac{1}{2(n-1)(4n-1)} \right)\right)}_{v_t^h \text{ if } h \notin t} \nonumber \\
    &= \frac{2n - 1}{8(n-1)n(4n-1)}
\end{align}
units of utility with respect to \eqref{eq: nodes baseline}.
Therefore, given \eqref{eq: nodes threshold} and \eqref{eq: nodes rate}, it must be that strictly more than
\begin{equation}
\label{eq: nodes contradiction}
    \frac{\frac{1}{2(n-1)(4n-1)}}{\frac{2n - 1}{8(n-1)n(4n-1)}} 
    = \frac{4n}{2n-1}
\end{equation}
units of $b(S)$ are allocated to goods $t \in T \setminus S$ with $h \in t$.

Now, consider the members of $S \cap T$.
We similarly establish a baseline for their collective utility.
If all of $b(S)$ were allocated to goods $S \cap T$, then their collective utility would be given by
\begin{equation}
\label{eq: edges baseline}
    \underbrace{4n}_{b(S)} \cdot \underbrace{\frac{n}{4n - 1}}_{v_t^t \text{ if } t \in S \cap T} = \frac{4n^2}{4n - 1}.
\end{equation}
On the other hand, since the members of $S \cap T$ form part of a blocking coalition against $u^*$, it must be that their collective utility strictly exceeds $|S \cap T| = n$.
Therefore, it must be that the members of $S \cap T$ collectively lose strictly less than
\begin{equation}
\label{eq: edges threshold}
    \frac{4n^2}{4n-1} - n = \frac{n}{4n-1}
\end{equation}
units of utility with respect \eqref{eq: edges baseline}.
Next, note that for each unit of budget allocated to goods $t \in T \setminus S$ with $h \in t$ rather than goods $S \cap T$, the members of $S \cap T$ collectively lose
\begin{equation}
\label{eq: edges rate}
    \underbrace{\frac{n}{4n-1}}_{v_t^t \text{ if } t \in S \cap T}  - \underbrace{0}_{v_{t'}^t \text{ if } t \in S \cap T, t' \in T \setminus S} = \frac{n}{4n-1}
\end{equation}
units of utility with respect to \eqref{eq: edges baseline}.
Therefore, given \eqref{eq: edges threshold} and \eqref{eq: edges rate}, it must be that strictly less than
\begin{equation}
\label{eq: edges contradiction}
    \frac{\frac{n}{4n-1}}{\frac{n}{4n - 1}} = 1
\end{equation}
unit from $b(S)$ is allocated to goods $T \setminus S$ with $h \in t$.
Finally, the budget allocation requirements in \eqref{eq: nodes contradiction} and \eqref{eq: edges contradiction} are in contradiction since $1 < \frac{4n}{2n - 1}$ for all $n \in \mathbb{N}$.
\end{proof}

\subsection*{Proof of Lemma~\ref{lemma: equivalence}}
\begin{proof}{Proof}
In this proof we consider \eqref{eq: C(N)}, as well as \eqref{eq: C'(N)} in its set difference form.

First, let $u^* \in C(N)$.
Then, $u^* \in U(N) = \projection{u}{Z(N)}$, so that there exists $x^* \in \R^J$ such that $(x^*, u^*) \in Z(N)$.
Moreover, since $u^* \in C(N)$, we have that $u^* \notin \interior{U(S)}$ for all $\emptyset \neq S \subseteq N$.
By Lemma~\ref{lemma: U - U'}, this implies that $(x^*, u^*) \notin \interior{U'(S)}$ for all $\emptyset \neq S \subseteq N$.
This shows that $(x^*, u^*) \in C'(N)$, so that $u^* \in \projection{u}{C'(N)}$

Conversely, let $u^* \in \projection{u}{C'(N)}$.
Then, there exists $x^* \in \R^J$ such that $(x^*, u^*) \in C'(N)$.
This implies that $(x^*, u^*) \in Z(N)$, so that $u^* \in \projection{u}{Z(N)} = U(N)$.
Moreover, since $(x^*, u^*) \in C'(N)$, we have that $(x^*, u^*) \notin \interior{U'(S)}$ for all $\emptyset \neq S \subseteq N$.
By Lemma~\ref{lemma: U - U'}, this implies that $u^* \notin \interior{U(S)}$ for all $\emptyset \neq S \subseteq N$.
This shows that $u^* \in C(N)$.
\end{proof}

\subsection*{Proof of Lemma~\ref{lemma: conv(P'-int(U'(S))}}
\begin{proof}{Proof}
Consider any $\emptyset \neq S \subseteq N$ and any extreme point $(x^*, u^*) \in P'$ satisfying $(x^*, u^*) \in \interior{U'(S)}$.

By way of contradiction, suppose $(x^*, u^*) \in \conv{P' \setminus \interior{U'(S)}}$.
Then, there exist $1 \leq \ell \leq J + N$ distinct points $(x^1, u^1), (x^2, u^2), \ldots, (x^\ell, u^\ell) \in P' \setminus \interior{U'(S)}$ and positive weights $\lambda^1, \lambda^2, \ldots, \lambda^\ell > 0$ such that $\sum_{w = 1}^\ell \lambda^w = 1$ and $(x^*, u^*) = \sum_{w = 1}^\ell \lambda^w (x^w, u^w)$.
If $\ell = 1$, then $(x^*, u^*) = (x^1, u^1) \in P' \setminus \interior{U'(S)}$.
This contradicts the assumption that $(x^*, u^*) \in \interior{U'(S)}$.
Therefore, suppose $\ell > 1$.
Note that $(x^1, u^1), (x^2, u^2), \ldots, (x^\ell, u^\ell) \in P' \setminus \interior{U'(S)}$ implies $(x^1, u^1), (x^2, u^2), \ldots, (x^\ell, u^\ell) \in P'$.
In particular, there exist $1 < \ell \leq J + N$ distinct points $(x^1, u^1), (x^2, u^2), \ldots, (x^\ell, u^\ell) \in P'$ and positive weights $\lambda^1, \lambda^2, \ldots, \lambda^\ell > 0$ such that $\sum_{w = 1}^\ell \lambda^w = 1$ and $(x^*, u^*) = \sum_{w = 1}^\ell \lambda^w (x^w, u^w)$.
This contradicts the assumption that $(x^*, u^*) \in P'$ is an extreme point.
\end{proof}

\subsection*{Proof of Lemma~\ref{lemma: valid}}
\begin{proof}{Proof}
Note that
\begin{align*}
    \conv{C'(N)} 
    &= \conv{\bigcap_{\emptyset \neq S \subseteq N}  P' \setminus\interior{U'(S)}} \\
    &\subseteq \bigcap_{\emptyset \neq S \subseteq N} \conv{P' \setminus\interior{U'(S)}}.
\end{align*}
Therefore, for any $\emptyset \neq S \subseteq N$, a valid inequality for $\conv{P' \setminus \interior{U'(S)}}$ is itself a valid inequality for $\conv{C'(N)}$.
\end{proof}

\subsection*{Additional Figures for Section~\ref{sec: frequency setting in the chicago bus system}}
\begin{figure}[ht]
    \centering
    \caption{
        Basis condition numbers observed throughout the execution of Algorithm~\ref{alg: algorithm}.
    }
    \label{fig: kappa}
    \includegraphics[trim={0cm 0cm 2cm 2.5cm},clip,width=0.825\linewidth]{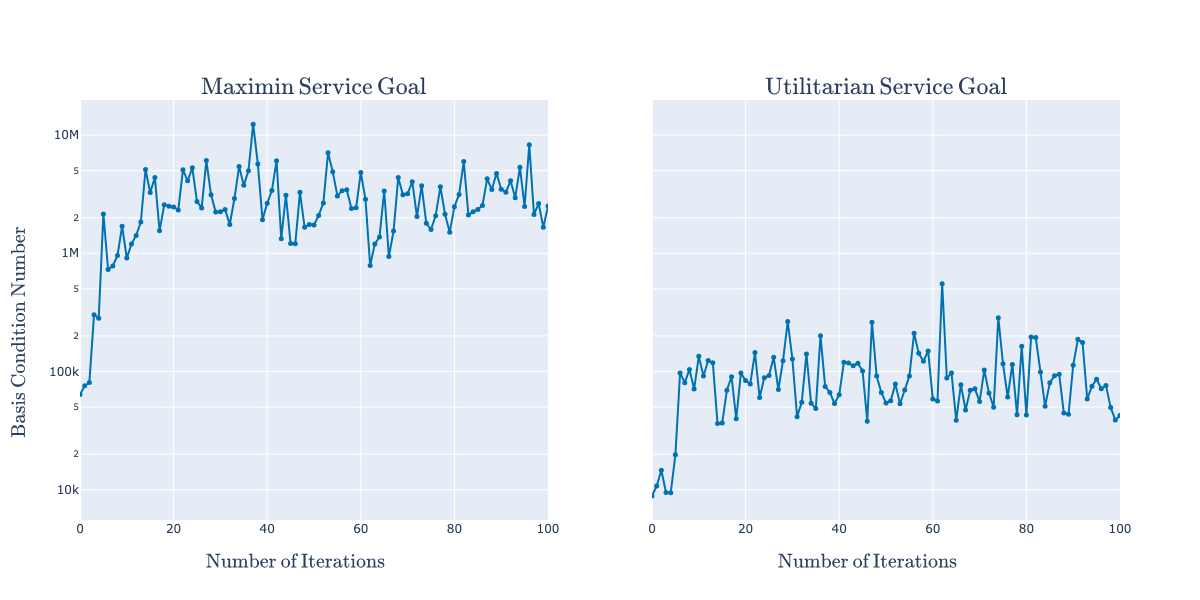}
\end{figure}

\begin{figure}[ht]
    \centering
    \caption{
        Optimality gap at 5 minute timeout of multiplicative least objections observed throughout the execution of Algorithm~\ref{alg: algorithm}.
        We ran our experiments on a \texttt{MacBook Air} with an \texttt{Apple M2} chip and 8 GB of memory.
    }
    \label{fig: eps mip gap}
    \includegraphics[trim={0cm 0cm 2cm 2.5cm},clip,width=0.825\linewidth]{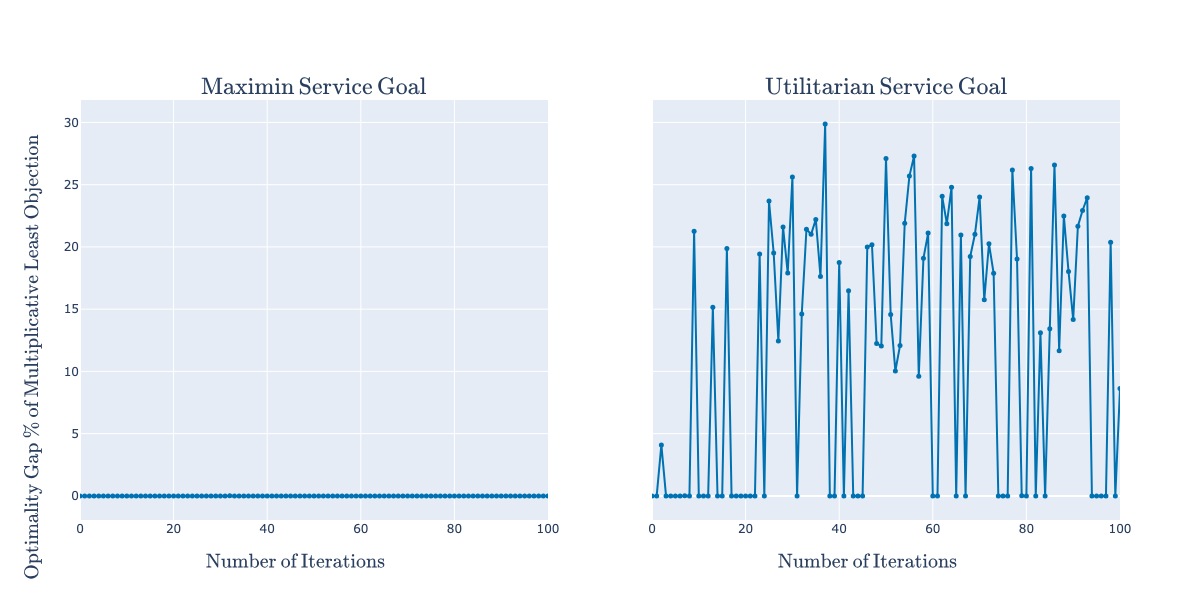}
\end{figure}
\begin{figure}[ht]
    \centering
    \caption{
    Elapsed time of the execution of Algorithm~\ref{alg: algorithm} under a 5 minute timeout for the multiplicative least objection formulation.
    We ran our experiments on a \texttt{MacBook Air} with an \texttt{Apple M2} chip and 8 GB of memory.
    As evidenced in Figure~\ref{fig: eps mip gap}, under the maximin service goal, the multiplicative least objection formulation is always solved to optimality within the timeout.
    This is in turn reflected as a much lower elapsed time under the maximin service goal than under the utilitarian service goal for the same number of iterations.
    }
    \label{fig: tot time}
    \includegraphics[trim={0cm 0cm 2cm 2.5cm},clip,width=0.825\linewidth]{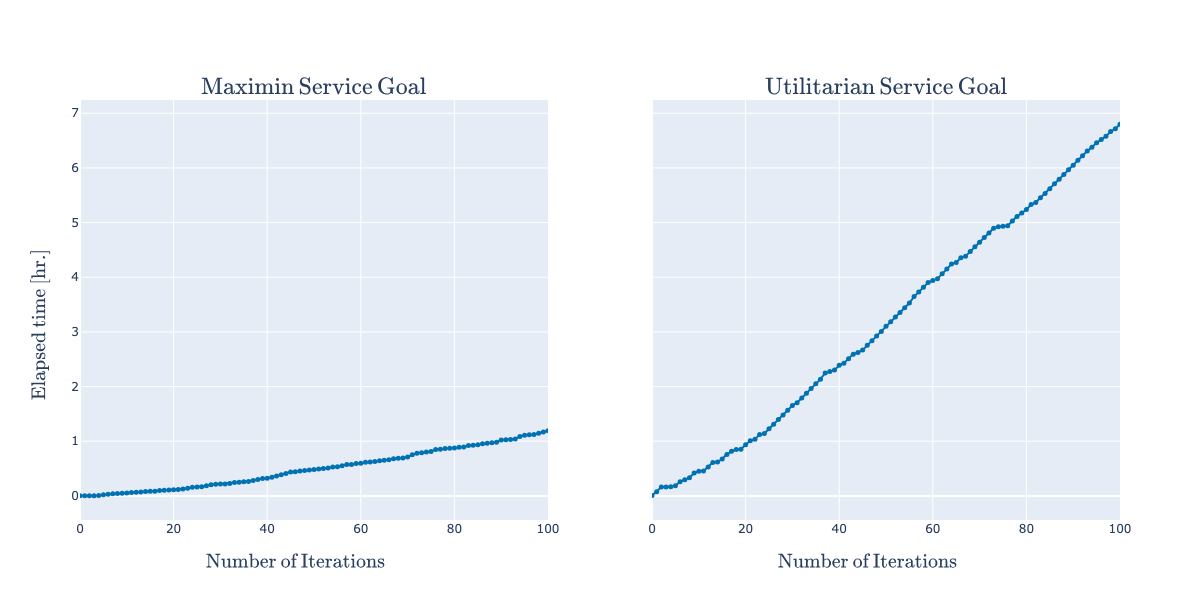}
\end{figure}

\end{document}